\documentclass[journal]{IEEEtran}
\usepackage{cite}
\usepackage{amsmath}
\usepackage{amssymb}
\usepackage{amsfonts}
\usepackage{cases}
\usepackage{graphicx}
\usepackage{color}
\usepackage{multirow}
\usepackage{booktabs}
\usepackage{relsize}
\usepackage{stfloats}
\usepackage{xcolor}
\usepackage{enumerate}
\usepackage{float}
\usepackage{setspace}
\usepackage{verbatim}
\usepackage{bm}
\usepackage{xfrac}
\usepackage{subcaption}
\usepackage{dsfont}
\usepackage{lipsum}
\usepackage{hyperref}
\usepackage{algorithm}
\usepackage{algpseudocode}

\usepackage{pgfplots}
\pgfplotsset{compat=1.17} 

\newtheorem{remark}{Remark}

\newtheorem{proposition}{Proposition}

\newtheorem{corollary}{Corollary}

\newcommand{\rom}[1]{\uppercase\expandafter{\romannumeral #1\relax}}

\allowdisplaybreaks

\begin{document}

\title{Quantized Zero-Energy RIS: Residual Phase Modeling and Outage Analysis}
\author{Dimitrios~Tyrovolas,~\IEEEmembership{Member,~IEEE}, Sotiris A. Tegos,~\IEEEmembership{Senior Member,~IEEE}, Kunrui Cao,~\IEEEmembership{Senior Member,~IEEE}, Yue Xiao,~\IEEEmembership{Member,~IEEE,} Panagiotis D. Diamantoulakis,~\IEEEmembership{Senior Member,~IEEE}, \\ Nikos C. Sagias,~\IEEEmembership{Senior Member,~IEEE}, Stylianos D. Asimonis,~\IEEEmembership{Senior Member,~IEEE}, \\ Christos K. Liaskos and George K.~Karagiannidis,~\IEEEmembership{Fellow,~IEEE}

\thanks{ 
D. Tyrovolas, S. A. Tegos, P. D. Diamantoulakis, and G. K. Karagiannidis are with the Department of Electrical and Computer Engineering, Aristotle University of Thessaloniki, 54124 Thessaloniki, Greece (e-mail: tyrovolas@auth.gr, tegosoti@auth.gr, padiaman@auth.gr, geokarag@auth.gr).

K. Cao is with Information Support Force Engineering University, Wuhan 430035, China, and also with the School of Information and Communications, National University of Defense Technology, Wuhan 430035, China (e-mail: krcao@nudt.edu.cn).

Y. Xiao is with the School of Information Science and Technology, Southwest Jiaotong University, 610031, Chengdu, China (e-mail: alice\_xiaoyue@hotmail.com)

N. C. Sagias is with the Department of Informatics and Telecommunications, University of Peloponnese, 22131 Tripoli,
Greece (e-mail: nsagias@uop.gr).

S. D. Asimonis is with the Department of Electrical and Computer Engineering, University of Patras, 26504 Patras, Greece (e-mail: s.asimonis@upatras.gr).

C. K. Liaskos is with the Computer Science Engineering Department, University of Ioannina, Ioannina, and Foundation for Research and Technology Hellas (FORTH), Greece (e-mail: cliaskos@ics.forth.gr).

}
}
\maketitle

\begin{abstract}
Zero-energy reconfigurable intelligent surfaces (zeRISs) have recently emerged as a promising solution for enabling energy-efficient and scalable programmable wireless environments (PWEs) by harvesting their operational energy from impinging radio-frequency signals. However, the operation of zeRIS-assisted systems is inherently constrained by the coupling between energy harvesting and signal reflection, a dependency that becomes more intricate under practical hardware limitations such as finite-resolution phase control. In this paper, we develop a comprehensive analytical framework for zeRIS-assisted communication systems operating under quantized phase shifts and harvest-and-reflect (HaR) schemes. Specifically, we analyze the joint energy–data rate outage probability and the energy efficiency under time switching and element splitting schemes, considering both transmitter-side and user-side deployment scenarios. By explicitly modeling the residual phase error induced by quantization and incorporating its statistical properties into the analysis, we show that quantization jointly affects energy harvesting and signal reflection, thereby inducing non-trivial trade-offs. As a result, the presented framework enables accurate performance evaluation and reveals critical design trade-offs for the selection of the phase resolution, and the applied HaR scheme in zeRIS-assisted wireless networks.
\end{abstract}

\begin{IEEEkeywords}
Zero-energy Reconfigurable Intelligent Surfaces (zeRIS), Phase Quantization, Residual Phase Error, Energy Efficiency, Performance Analysis
\end{IEEEkeywords}

\section{Introduction}

In next-generation wireless networks, emerging applications such as immersive extended reality and the metaverse impose stringent requirements on reliability, latency, and data rate, while simultaneously demanding that the wireless environment adapts to diverse user objectives \cite{Holographic, 6GMetaverse}. To meet these requirements, programmable wireless environments (PWEs) have been introduced as a paradigm in which electromagnetic wave propagation is treated as a controllable and software-defined process \cite{Tutorial,PWELiaskos}. In more detail, the PWE paradigm relies on the deployment of distributed reconfigurable nodes that dynamically adjust their electromagnetic response and collectively shape the propagation conditions across the environment. However, since effective control of the wireless channel requires the manipulation of both large-scale and small-scale propagation effects, achieving the desired performance inherently necessitates a dense deployment of such nodes, which in turn introduces critical challenges in terms of energy consumption and scalability \cite{PWELiaskos}. In such settings, sustaining the operation of a large number of reconfigurable nodes through conventional power supply mechanisms becomes impractical, necessitating alternative solutions based on self-sustained operation. This leads to the need for reconfigurable nodes that can operate as zero-energy devices (ZEDs), harvesting energy from ambient radio-frequency signals while maintaining their capability to dynamically control the wireless propagation environment \cite{zed}. Therefore, the realization of PWE nodes that operate in a self-sustainable manner while preserving their ability to control the wireless propagation environment becomes essential for enabling scalable PWEs where conventional power supply mechanisms fail to sustain continuous operation.

Among the technologies envisioned to support energy-efficient PWEs, zero-energy reconfigurable intelligent surfaces (zeRISs) have emerged as a promising approach to enable autonomous and energy-efficient operation \cite{zeris}. By integrating energy harvesting mechanisms into the metasurface architecture and exploiting their absorption functionality \cite{pitilakis,SAMKunrui}, zeRISs are capable of harvesting the required operational energy directly from impinging radio-frequency signals, thus eliminating the need for dedicated power sources \cite{zeris, wirelesspoweredris}. At the same time, through appropriate configuration of their reflecting elements, zeRISs retain the ability to dynamically control the wireless propagation environment and realize functionalities such as beam steering and beam splitting. This capability is supported at the element level, where a portion of the incident electromagnetic power is absorbed and converted into electrical energy to support the controller and the impedance tuning circuitry, while the remaining portion is utilized to manipulate the reflected wave according to the PWE objectives \cite{kunruizeris}. As a result, the operation of zeRISs relies on a shared use of the incident signal that simultaneously supports both energy harvesting and wave manipulation, rendering the amount of harvested energy and the effectiveness of the applied electromagnetic functionality inherently coupled \cite{zeris}. Consequently, understanding the extent to which zeRISs can sustain their operation while delivering the desired PWE services becomes critical, motivating the need to systematically characterize their achievable performance.

\subsection{State-of-the-art}
An increasing amount of literature has focused on understanding how reconfigurable intelligent surfaces (RISs) can operate in a self-sustainable manner, thereby establishing the foundations of zeRIS as autonomous components of PWEs \cite{zeris,comparing}. In this direction, the authors of \cite{kunruizeris} analyzed a zeRIS-assisted network from the perspectives of security, reliability, and energy efficiency, thus showing that the zeRIS concept can be embedded into meaningful communication-theoretic tradeoffs beyond mere feasibility. Along the same line, \cite{DingZeris} studied a zeRIS-assisted communication architecture with noise modulation and interference-based energy harvesting, revealing that zeRIS operation can also be supported through more sophisticated signal and absorption mechanisms. Furthermore, the authors of \cite{LUYAO} considered a cooperative zero-energy double-RIS architecture for wireless power transfer, indicating that the zeRIS paradigm can be extended to multi-surface scenarios, thus further supporting its applicability as a PWE component. Building on these developments, the authors of \cite{OCTAVIAZERIS} investigated a NOMA-based zeRIS-empowered backscatter communication system with energy-efficient resource management, demonstrating that zeRIS operation can be integrated with practical resource allocation, amplitude reflection control, and energy-harvesting constraints in more advanced communication architectures. Finally, recent works such as \cite{SRIS, ARES} further support the practical feasibility of zeRIS operation by demonstrating RF-powered reconfigurable surfaces. Therefore, these studies collectively show that zeRISs have already evolved from a feasibility-oriented concept into a flexible and practically viable communication framework capable of supporting diverse operating objectives and architectural realizations.

The design and performance implications of zeRIS-assisted systems have also been examined from a resource allocation and system optimization perspective, providing further insight into their behavior under practical constraints. In more detail, the authors of \cite{ntontin1, Chu2022,VehicularZeris} examined the feasibility of enabling autonomous zeRIS operation through energy harvesting, thereby establishing the fundamental conditions under which zeRISs can operate without dedicated power sources. Building on this foundation, the authors of \cite{MECZeris,luyao2,robust,Pan2022,Yang2025} investigated zeRIS-assisted wireless networks under performance-oriented design objectives, showing that energy constraints directly influence beamforming design, transmit power allocation, and overall system performance across different communication scenarios. Moreover, the authors of \cite{Laue2023} showed that zeRIS implementation also affects practical procedures such as beam acquisition and configuration, while \cite{LaueCodebook, YuZheng2023} demonstrated that the achievable performance depends on the adopted energy management strategy, including different signal-sharing mechanisms and operating conditions. Taken together, these studies indicate that the achievable performance of zeRISs depends not only on the PWE objectives, but also on how the impinging signal is manipulated to support both energy harvesting and electromagnetic functionality across different propagation environments.

\subsection{Motivation \& Contribution}

The practical realization of zeRIS-assisted systems requires that the electromagnetic functionality imposed by the surface is implemented through hardware-constrained control mechanisms. In particular, the configuration of zeRIS elements is achieved through discrete phase shifts, where the applied values are quantized and determined by the underlying circuit design \cite{Hongliang2020,badiu2020communication}. This discretization limits the ability of the zeRIS to shape the wireless propagation environment, as the imposed phase configuration deviates from the ideal continuous control typically assumed in analysis. At the same time, higher-resolution control mechanisms incur increased power consumption at the element level, directly affecting the energy available to sustain zeRIS operation \cite{LaueCodebook}. As a result, the practical implementation of zeRIS functionality gives rise to a non-trivial interaction between electromagnetic performance and energy consumption, which becomes more intricate under different propagation conditions and deployment scenarios. In particular, the effect of phase quantization may vary significantly across channel conditions, where inaccurate characterization of its impact can lead to misleading performance predictions, and consequently, non-optimal zeRIS design. Therefore, to the best of the authors’ knowledge, no existing work has explicitly characterized the impact of phase quantization on zeRIS-assisted systems by jointly accounting for both energy sustainability and programmable wireless functionality under realistic propagation conditions.

In this paper, we investigate the performance of a zeRIS-assisted communication system under practical hardware-constrained operation, where the applied phase configuration is affected by quantization. In particular, we consider the \emph{time switching (TS)} and \emph{element splitting (ES)} schemes, and examine both \emph{transmitter (Tx)-side} and \emph{user equipment (UE)-side} zeRIS deployment scenarios. The main contributions of this work can be summarized as follows:

\begin{itemize}
\item We develop a mathematical framework for zeRIS-assisted systems with quantized phase shifts by explicitly modeling the resulting residual phase error and incorporating its statistical properties into the performance analysis. Based on this framework, we derive closed-form expressions for the joint energy–data rate outage probability under the considered HaR schemes and deployment scenarios, enabling the characterization of the interplay between energy sustainability and communication under quantized operation.
\item We investigate the impact of phase quantization on zeRIS performance across different propagation conditions, HaR schemes, and deployment scenarios, showing that the effect of quantization is strongly dependent on the channel characteristics and the relative positioning of the zeRIS in the network.
\item We reveal non-trivial design trade-offs induced by quantized operation, demonstrating that increasing the phase resolution does not necessarily improve performance due to the associated increase in energy consumption, while also showing that simplified or inaccurate characterization of quantization effects may lead to ineffective zeRIS designs.
\end{itemize}

\subsection{Structure}
The remainder of the paper is organized as follows. Section II presents the system model, including the considered HaR schemes and deployment scenarios. Section III develops the statistical characterization of the residual phase error, and provides the performance metrics. Section IV presents the numerical results, and finally Section V concludes the paper.

\section{System Model}

We consider a downlink system where a single-antenna Tx communicates with a single-antenna UE through a zeRIS equipped with \(N\) passive reflecting elements. The received signal is obtained as the coherent superposition of the signals reflected by the zeRIS and is given by
\begin{equation}
\small
y = \sqrt{\ell_p P_t G} \sum_{i=1}^N |h_{1i}| |h_{2i}| e^{j\phi_i} x + n,
\end{equation}
where \(x\) denotes the transmitted symbol with unit power, \(P_t\) is the transmit power, \(G = G_t G_r\) is the combined antenna gain, and \(n \sim \mathcal{CN}(0,\sigma^2)\) represents the additive noise. Additionally, \(h_{1i}\) and \(h_{2i}\) denote the channel coefficients of the Tx-to-zeRIS and zeRIS-to-UE links associated with the \(i\)-th reflecting element, respectively, \(\ell_p = \ell_1 \ell_2\) captures the end-to-end path loss and is equal to 
\begin{equation}
\small
\ell_u = C_0 d_u^{-a_u}, \quad u \in \{1,2\},
\end{equation}
where \(d_1\) denotes the Tx-zeRIS distance, \(d_2\) denotes the zeRIS-UE distance, \(a_1\) and \(a_2\) are the corresponding path loss exponents, and \(C_0 = \lambda^2 / (16 \, \pi^2)\) is the path loss constant determined by the carrier wavelength \(\lambda\). Moreover, the phase term \(\phi_i\) accounts for both the propagation and the controllable reflection effects and is given by \(\phi_i = \omega_i + \arg(h_{1i}) + \arg(h_{2i})\), where \(\arg(\cdot)\) denotes the phase of a complex channel coefficient. Finally, \(\omega_i\) represents the phase induced by the zeRIS and is selected from the finite set \(\mathcal{Q} = \{0, \Delta, \dots, (2^q-1)\Delta\}\), with \(\Delta = {2\pi}/{2^q}\) and \(q\) denoting the phase resolution in bits with \(2^q\) discrete phase states, through a quantization mapping \(\omega_i = \mathbb{Q}(\cdot)\) applied to the desired continuous phase. Consequently, the received signal is directly shaped by the joint effect of the wireless channel phases and the discrete control imposed by the zeRIS.

Building on this signal model, the zeRIS operates as a self-sustained programmable surface that relies solely on the impinging electromagnetic waves to support its functionality. In particular, each reflecting element is equipped with circuitry that enables both energy harvesting and controllable reflection, such that part of the received power is converted into electrical energy and used to configure the reflecting elements and operate the zeRIS controller, which coordinates the phase profile across the surface \cite{zeris}. To facilitate the initiation of this process, each element is assumed to exhibit a minimal activation capability, which can be realized through an initial stored charge or through passive biasing induced by the impinging electromagnetic field, thereby enabling the energy harvesting circuitry to operate without an external power supply. In this context, the total energy consumption depends on the number of elements and their hardware characteristics. Specifically, each reflecting element consumes a power denoted by \(P_{\mathrm{elem}}\), which is determined by the phase resolution \(q\), since the realization of discrete phase shifts requires switching components implemented with PIN diodes, where each additional quantization bit introduces an additional diode to represent the corresponding phase states. As a result, the element power consumption scales linearly with the phase resolution and can be expressed as \(P_{\mathrm{elem}} = q P_{\mathrm{PIN}}\), assuming identical phase resolution across all elements, where \(P_{\mathrm{PIN}}\) denotes the power consumption of a single diode \cite{alexandropoulosee}. In addition, a constant power term denoted by \(P_{\mathrm{ctrl}}\) is required to support the operation of the control circuitry. Therefore, the total energy requirement is determined by the elements that participate in the imposed configuration together with the controller consumption and must be fully supported by the harvested energy. As a result, the ability of the zeRIS to effectively shape the wireless channel is inherently constrained by the available energy, establishing a direct coupling between the propagation environment, the surface configuration, and the achievable system performance.

\subsection{zeRIS Deployment}
The inherent dependence of both the reflected signal and the harvested energy on the propagation conditions highlights the importance of the zeRIS spatial configuration in determining the statistical behavior of the system. In PWEs, where multiple surfaces are deployed to shape the propagation in a coordinated manner, the placement of each surface is primarily dictated by the double path loss of the cascaded channel. Since the overall attenuation is minimized when either the Tx-zeRIS or the zeRIS-UE distance is small, practical deployments tend to position the zeRIS either close to the Tx or in proximity to the users, particularly in dense scenarios where reliable coverage is required. In more detail, when the zeRIS is located near the Tx, i.e., Tx-side deployment, the Tx-zeRIS link is dominated by a line-of-sight (LoS) component with $|h_{1i}|=1$ and $\arg(h_{1i})=\mod(\frac{2\pi d_1}{\lambda}, 2\pi)$, while the zeRIS-UE link experiences fading. Conversely, when the zeRIS is placed near the UE, i.e., UE-side deployment, the zeRIS-UE link becomes LoS with $|h_{2i}|=1$ and $\arg(h_{2i})=\mod(\frac{2\pi d_2}{\lambda}, 2\pi)$, whereas the Tx-zeRIS link is subject to fading. In both cases, the amplitude of the fading link is modeled as a Nakagami-$m$ random variable (RV) with probability density function
\begin{equation}
\small
f_{|h|}(x)=\frac{2m_n^{m_n}}{\Gamma(m_n)\Omega^{m_n}}x^{2{m_n}-1}\exp\left(-\frac{m_n}{\Omega}x^2\right), \quad x\geq 0,
\end{equation}
where $m_n$ and $\Omega$ denote the shape and spread parameters, respectively \cite{TEGOSTVT}. Moreover, the phase of the fading link is described by \cite{Abdi}
\begin{equation} \label{phase_distribution}
\small
f_p(\theta) = \frac{e^{\kappa \cos\left(\theta - \frac{2\pi d_u}{\lambda}\right)}}{2 \pi (K+1) I_0(\kappa)} + \frac{K}{K+1} \delta\left(\theta - \frac{2\pi d_u}{\lambda}\right),
\end{equation}
where $u\in\{1,2\}$ depending on the link under consideration, $\delta(\cdot)$ is the Dirac delta function, $\kappa \geq 0$ is the concentration parameter that controls the dispersion of the phase around its mean value, with $\kappa=0$ corresponding to a uniform phase distribution and larger values indicating a stronger phase concentration, $I_0(\cdot)$ is the zero-order modified Bessel function of the first kind, and $K \geq 0$ denotes the Rice factor. In particular, by relating the fading severity to $m_n$, $K$ can be expressed as \cite{kappa}
\begin{equation}
\small
K \approx \frac{\sqrt{m^2-m}}{m-\sqrt{m^2-m}}.
\end{equation}
In this way, the two deployment regimes provide a tractable and physically meaningful basis for modeling both the amplitude and phase statistics of the cascaded channels.

\subsection{Harvest-and-Reflect (HaR) Schemes}

The coupled dependence of the reflected signal and the harvested energy on the channel phases implies that the zeRIS operation depends on the applied phase configuration, which characterizes how the impinging field is utilized for energy acquisition and signal reflection. In particular, aligning the phases with the Tx-zeRIS link maximizes the harvested energy, leading to $\omega_i^{\mathrm{h}} = \mathbb{Q}(-\arg(h_{1i}))$, whereas aligning the phases over the cascaded channel enables coherent signal combining at the receiver, yielding $\omega_i^{\mathrm{c}} = \mathbb{Q}(-\arg(h_{1i}) - \arg(h_{2i}))$. This interplay between energy harvesting and signal reflection naturally gives rise to harvest-and-reflect (HaR) schemes, in which the zeRIS operation is structured around these two competing objectives, motivating the schemes examined in this work, namely TS and ES.

\subsubsection{Time Switching (TS)}
In the TS scheme, the zeRIS alternates between energy harvesting and signal reflection over the transmission interval, such that the block of duration $T$ is divided into two portions, where a fraction $\tau \in [0,1]$ is allocated to energy harvesting and the remaining fraction is used for information transmission. During the information transmission phase, the achievable rate is given by
\begin{equation}
\small
\begin{split}
&R_{\mathrm{TS}}= (1-\tau) \\
&\times\log_2\Bigg(1 \!+\! \gamma_t G \ell_p \left| \sum_{i=1}^{N} |h_{1i}| |h_{2i}| e^{j(\omega_i^{\mathrm{c}} + \arg(h_{1i}) + \arg(h_{2i}))} \right|^2 \Bigg) ,
\end{split}
\end{equation}
where $\gamma_t = \frac{P_t}{\sigma^2}$. During the energy harvesting phase, the reflecting elements are configured for energy acquisition, yielding
\begin{equation}
\small
Q_{\mathrm{TS}} = \tau T \zeta P_t G_t \ell_1 \left| \sum_{i=1}^{N} |h_{1i}| e^{j(\omega_i^{\mathrm{h}} + \arg(h_{1i}))} \right|^2,
\end{equation}
where $\zeta$ denotes the energy conversion efficiency, accounting for both the inherent inefficiencies of the harvesting circuitry and the power losses associated with its operation, and thus capturing the effective energy available to sustain the zeRIS functionality, while the corresponding energy consumption is given by
\begin{equation}
\small
E_{\mathrm{TS}} = T\left((1-\tau) N P_{\mathrm{elem}} + P_{\mathrm{ctrl}}\right).
\end{equation}

\subsubsection{Element Splitting (ES)}
In contrast to TS, the ES scheme distributes the energy harvesting and beam-steering processes across the reflecting elements, such that a subset of $N_1$ elements is dedicated to energy harvesting, while the remaining $N_2$ elements support information transmission over the entire interval, with $N_1 + N_2 = N$. Under this configuration, phase shifts are applied only to the elements participating in signal reflection, while the remaining elements are configured for energy harvesting. Accordingly, the received signal is formed by the transmitting subset as
\begin{equation}
\small
y_{\mathrm{ES}} = \sqrt{\ell_p P_t G} \sum_{i=1}^{N_2} |h_{1i}| |h_{2i}| e^{j(\omega_i^{\mathrm{c}} + \arg(h_{1i}) + \arg(h_{2i}))} x + n,
\end{equation}
and the corresponding achievable rate is given by
\begin{equation}
\small
R_{\mathrm{ES}}\! =\! \log_2\!\left(1 \!+ \!\gamma_t G \ell_p \left| \sum_{i=1}^{N_2} |h_{1i}| |h_{2i}| e^{j(\omega_i^{\mathrm{c}} + \arg(h_{1i}) + \arg(h_{2i}))} \right|^2 \right).
\end{equation}
Meanwhile, the harvested energy is obtained from the complementary subset as
\begin{equation}
Q_{\mathrm{ES}} = T \zeta P_t G_t \ell_1 \left| \sum_{i=1}^{N_1} |h_{1i}| e^{j(\omega_i^{\mathrm{h}} + \arg(h_{1i}))} \right|^2,
\end{equation}
while the corresponding energy consumption is given by
\begin{equation}
E_{\mathrm{ES}} = T\left(N_2 P_{\mathrm{elem}} + P_{\mathrm{ctrl}}\right).
\end{equation}

\section{Statistical Characterization of the Residual Phase Error}

Building on the considered system model, the operation of the zeRIS critically depends on the phase shifts applied across its reflecting elements. However, due to finite-resolution phase control, the applied phase cannot perfectly match its ideal value, which gives rise to a residual phase error. Since this residual phase error directly affects both the harvested energy and the coherent combination of the reflected signals, its statistical characterization is essential for the analysis of zeRIS with quantized phase shifts. Therefore, in this section, we therefore derive the probability density function (PDF) of the residual phase error and then establish its circular moments, which constitute the key statistical quantities required for the analytical treatment of the considered zeRIS-assisted network. In this direction, the PDF of the residual phase error under the considered fading model is presented in the following proposition.
\begin{proposition}
The PDF of the residual phase error $\varepsilon$ induced by the $q$-bit quantization can be expressed as
\begin{equation}\label{error_distribution}
\small
\begin{split}
f_{\varepsilon}(\varepsilon)
&=
\frac{1}{(K+1)2\pi I_0(\kappa)}
\sum_{m=0}^{2^q-1}
\exp\!\Big(\kappa \cos\!\big(m\Delta-\varepsilon+\tfrac{2\pi d_u}{\lambda}\big)\Big) \\
&\quad + \frac{K}{K+1}\,\delta(\varepsilon-\varepsilon_d), \quad  \varepsilon \in [-\Delta/2,\Delta/2],
\end{split}
\end{equation}
where $\varepsilon_d=\operatorname{mod}\!\left(\mathbb{Q}\!\left(-\tfrac{2\pi d_u}{\lambda}\right)+\tfrac{2\pi d_u}{\lambda}
+\tfrac{\Delta}{2},\,\Delta\right)-\tfrac{\Delta}{2}$.
\end{proposition}
\begin{IEEEproof}
Let $\phi$ denote the desired continuous phase that is quantized through the mapping $\omega_i=\mathbb{Q}(\phi)$, where $\mathbb{Q}(\cdot)$ selects the closest element from the finite set $\mathcal{Q}=\{0,\Delta,\dots,(2^q-1)\Delta\}$. By definition, the residual phase error is given by
\begin{equation}\label{eq:eps_def}
\small
\varepsilon=\omega_i-\phi.
\end{equation}
Since $\mathbb{Q}(\cdot)$ maps $\phi$ to its nearest quantization level, it follows directly from \eqref{eq:eps_def} that $\varepsilon \in [-\Delta/2,\Delta/2]$.

For a given quantization level $m \, \Delta$, with $m \in \{0,\dots,2^q-1\}$, the condition $\mathbb{Q}(\phi)=m\Delta$ implies $\phi = m \, \Delta - \varepsilon$. Moreover, since the quantization mapping partitions the phase space into $2^q$ disjoint regions associated with the reconstruction levels $\{m \, \Delta\}$, the PDF of $\varepsilon$ is obtained by summing the contributions from all quantization levels, which yields
\begin{equation}\label{eq:fe_sum}
\small
f_{\varepsilon}(\varepsilon)
=
\sum_{m=0}^{2^q-1} f_{\phi}(m\Delta-\varepsilon).
\end{equation}
Therefore, by substituting \eqref{phase_distribution} into \eqref{eq:fe_sum}, we obtain
\begin{equation}\label{eq:fe_expanded}
\small
\begin{split}
f_{\varepsilon}(\varepsilon)
&=
\sum_{m=0}^{2^q-1}
\Bigg[
\frac{\exp\!\left(\kappa\cos\!\left(m\Delta-\varepsilon+\frac{2\pi d_u}{\lambda}\right)\right)}
{2\pi (K+1) I_0(\kappa)} \\
&\hspace{2.4cm}
+\frac{K}{K+1}\delta\!\left(m\Delta-\varepsilon+\frac{2\pi d_u}{\lambda}\right)
\Bigg].
\end{split}
\end{equation}
Finally, each Dirac term is non-zero only if its argument satisfies $\varepsilon= m\Delta+\frac{2\pi d_u}{\lambda}$. However, due to the quantization mapping, the deterministic phase $-\frac{2\pi d_u}{\lambda}$ is assigned to a unique quantization level $Q\!\left(-\frac{2\pi d_u}{\lambda}\right)= m_d \, \Delta$. Therefore, only the term corresponding to $m=m_d$ satisfies the above condition, while for all other values of $m$ the argument of the Dirac function is non-zero and hence those terms vanish. As a result, for $m=m_d$, the corresponding residual phase error follows directly from \eqref{eq:eps_def} as
\begin{equation}
\small
\varepsilon = \mathbb{Q}\!\left(-\tfrac{2\pi d_u}{\lambda}\right)
+\tfrac{2\pi d_u}{\lambda},
\end{equation}
and since the residual phase error is defined within the interval $[-\Delta/2, \, \Delta/2]$, the above expression is mapped to this interval through a centered modulo-$\Delta$ operation that shifts the value to $[-\Delta/2, \, \Delta/2]$, which yields $\varepsilon_d$. Therefore, \eqref{eq:fe_expanded} can be rewritten as in \eqref{error_distribution}, which concludes the proof.
\end{IEEEproof}

Following the PDF characterization of the residual phase error, it is important to examine the scenario where the channel affected by fading does not include a LoS component. In this case, the channel's amplitude is Rayleigh distributed and its phase is uniformly distributed over $[0, \,  2 \, \pi)$, which is reflected in the adopted phase model by setting $\kappa=0$ and $K=0$, removing the deterministic component \cite{TEGOSTVT}. Under these conditions, the residual phase error admits a simplified distribution, as stated in the following corollary.
\begin{corollary}
For $\kappa=0$ and $K=0$, the residual phase error $\varepsilon$ is uniformly distributed over the interval $[-\Delta/2,\Delta/2]$, i.e.,
$f_{\varepsilon}(\varepsilon)=\frac{1}{\Delta}$, where $\varepsilon \in \left[-\frac{\Delta}{2},\frac{\Delta}{2}\right]$.
\end{corollary}
\begin{IEEEproof}
By setting $\kappa=0$ and $K=0$ in \eqref{error_distribution}, we obtain
\begin{equation}\label{pdfunif}
\small
f_{\varepsilon}(\varepsilon) = \frac{2^q}{2\pi}.
\end{equation}
Finally, by using $\Delta=\frac{2\pi}{2^q}$, \eqref{pdfunif} can be rewritten as $\frac{1}{\Delta}$, which completes the proof.
\end{IEEEproof}

Building on the statistical characterization of the residual phase error, it is now essential to derive tractable quantities that can be directly utilized in the performance analysis of a zeRIS-assisted network. In this context, the circular moments capture the effect of phase quantization on the coherent combination of the reflected components and therefore provide a tractable way to connect the residual phase error model with the performance of a zeRIS-assisted network. To this end, we next derive the expression of the $n$-th circular moment of the residual phase error. 

\begin{proposition}
The \(n\)-th circular moment of the residual phase error \(\varepsilon\) is given by
\begin{equation}\label{eq:nth_moment}
\small
\begin{split}
\mu_n & (d_u)
=
\frac{1}{2\pi (K+1)}
\sum_{m=0}^{2^q-1}
\Bigg[
\frac{2\sin\!\left(\frac{n\Delta}{2}\right)}{n}  +
\sum_{\ell=1}^{+\infty}
\frac{I_\ell(\kappa)}{I_0(\kappa)}  \\&
\times \Bigg(
e^{j\ell\left(m\Delta+\frac{2\pi d_u}{\lambda}\right)}
\frac{2\sin\!\left(\frac{(n-\ell)\Delta}{2}\right)}{n-\ell}  +
e^{-j\ell\left(m\Delta+\frac{2\pi d_u}{\lambda}\right)}  \\&
\times \frac{2\sin\!\left(\frac{(n+\ell)\Delta}{2}\right)}{n+\ell}
\Bigg)
\Bigg]  + \frac{K}{K+1}e^{jn\varepsilon_d}.
\end{split}
\end{equation}
\end{proposition}

\begin{IEEEproof}
By definition, the \(n\)-th circular moment of the residual phase error is
\begin{equation}
\small
\mu_n=\mathbb{E}\!\left[e^{jn\varepsilon}\right]
=\int_{-\Delta/2}^{\Delta/2} e^{jn\varepsilon} f_{\varepsilon}(\varepsilon)\, d\varepsilon,
\end{equation}
where $\mathbb{E}[\cdot]$ denoting expectation \cite{badiu2020communication}. By substituting \eqref{error_distribution} into the above expression, we obtain
\begin{equation}\label{25}
\small
\begin{split}
&\mu_n=
\frac{1}{(K+1)2\pi I_0(\kappa)}
\sum_{m=0}^{2^q-1}
\int_{-\Delta/2}^{\Delta/2}
e^{jn\varepsilon}
e^{\kappa \cos\!\left(m\Delta-\varepsilon+\tfrac{2\pi d_u}{\lambda}\right)}
\, d\varepsilon \\
&\quad +
\frac{K}{K+1}
\int_{-\Delta/2}^{\Delta/2}
e^{jn\varepsilon}\delta(\varepsilon-\varepsilon_d)\, d\varepsilon .
\end{split}
\end{equation}
By utilizing the sifting property of the Dirac delta function and expanding $e^{\kappa \, \cos (x)}$ into its Fourier series form, \eqref{25} can be rewritten as
\begin{equation}\label{26}
\small
\begin{split}
\mu_n &
=
\frac{1}{(K+1)2\pi I_0(\kappa)}
\sum_{m=0}^{2^q-1}
\int_{-\Delta/2}^{\Delta/2}
e^{jn\varepsilon}
\Bigg[
I_0(\kappa) \\& +2\sum_{\ell=1}^{\infty} I_\ell(\kappa)
\cos\!\Big(\ell \big(m\Delta-\varepsilon+\tfrac{2\pi d_u}{\lambda}\big)\Big)
\Bigg]
d\varepsilon + \frac{K}{K+1}e^{jn\varepsilon_d}.
\end{split}
\end{equation}
Moreover, by utilizing the identity $\cos (x) =(e^{jx}+e^{-jx}) / 2$ and after some algebraic manipulations, \eqref{26} can be written as
\begin{equation}\label{27}
\small
\begin{split}
&\mu_n=
\frac{1}{2\pi (K+1)}
\sum_{m=0}^{2^q-1}
\Bigg[
\int_{-\Delta/2}^{\Delta/2} e^{jn\varepsilon}\, d\varepsilon  \\& \times \sum_{\ell=1}^{\infty}\!\!\frac{I_\ell(\kappa)}{I_0(\kappa)}
\Bigg(
\!e^{j\ell\left(m\Delta+\frac{2\pi d_u}{\lambda}\right)}
\!\!\!\int_{-\Delta/2}^{\Delta/2} \!\!\! e^{j(n-\ell)\varepsilon}d\varepsilon +\!e^{-j\ell\left(m\Delta+\frac{2\pi d_u}{\lambda}\right)} \\ &\times 
\int_{-\Delta/2}^{\Delta/2} e^{j(n+\ell)\varepsilon}\, d\varepsilon
\Bigg)
\Bigg]  + \frac{K}{K+1}e^{jn\varepsilon_d},
\end{split}
\end{equation}
where after evaluating the integrals in \eqref{27}, we obtain
\begin{equation}
\small
\begin{split}
&\mu_n=
\frac{1}{2\pi (K+1)}
\sum_{m=0}^{2^q-1}
\Bigg[
\frac{e^{jn\Delta/2}-e^{-jn\Delta/2}}{jn} +
\sum_{\ell=1}^{\infty}\frac{I_\ell(\kappa)}{I_0(\kappa)}  \\& \times
\Bigg(
e^{j\ell\left(m\Delta+\frac{2\pi d_u}{\lambda}\right)}
\frac{e^{j(n-\ell)\Delta/2}-e^{-j(n-\ell)\Delta/2}}{j(n-\ell)} \\
&+e^{-j\ell\left(m\Delta+\frac{2\pi d_u}{\lambda}\right)}
\frac{e^{j(n+\ell)\Delta/2}-e^{-j(n+\ell)\Delta/2}}{j(n+\ell)}
\Bigg)\!\Bigg] \!+ \!\frac{K}{K+1}e^{jn\varepsilon_d}.
\end{split}
\end{equation}
Finally, by utilizing the identity $e^{j \, x}-e^{-j \, x}=2 \, j \, \sin(x)$, we obtain \eqref{eq:nth_moment}, which concludes the proof.
\end{IEEEproof}

As can be observed, \eqref{eq:nth_moment} involves an infinite summation that arises from the Fourier-series expansion of \(e^{\kappa \, \cos (x)}\). In the following, we show that this series converges rapidly and can be accurately approximated using a finite number of terms, which is essential for the practical evaluation of the circular moments.

\begin{proposition}
Let \(\mu_n^{(L)}\) denote the \(n\)-th circular moment obtained from \eqref{eq:nth_moment} by truncating the infinite sum to $L$ terms. Then, the truncation error satisfies
\begin{equation}\label{eq:trunc_bound_exp}
\small
\left|\mu_n(d_u)-\mu_n^{(L)}(d_u)\right|
\leq
\frac{2\left[
\exp\!\left(\frac{\kappa}{2}\right)
-
\sum_{\ell=0}^{L}
\frac{1}{\ell!}\left(\frac{\kappa}{2}\right)^\ell
\right]}{(K+1)I_0(\kappa)},
\end{equation}
which monotonically decreases with $L$ and tends to 0 as $L \to \infty$.
\end{proposition}
\begin{IEEEproof}
Let \(\mu_n^{(L)}\) denote the truncated version of \eqref{eq:nth_moment}, obtained by retaining the first $L$ terms. Thus, the truncation error corresponds to the contribution of the neglected terms and can be bounded as
\begin{equation}
\small
\begin{split}
\left|\mu_n-\mu_n^{(L)}\right|
&\leq
\frac{1}{2\pi (K+1)}
\sum_{m=0}^{2^q-1}
\sum_{\ell=L+1}^{\infty}
\frac{I_\ell(\kappa)}{I_0(\kappa)} \\
&\quad \times
\Bigg(
\left|
\frac{2\sin\!\left(\frac{(n-\ell)\Delta}{2}\right)}{n-\ell}
\right|
+
\left|
\frac{2\sin\!\left(\frac{(n+\ell)\Delta}{2}\right)}{n+\ell}
\right|
\Bigg).
\end{split}
\end{equation}
Moreover, by using the bound $\left|
\frac{2\sin\!\left(\frac{a\Delta}{2}\right)}{a}
\right|
\leq \Delta$, and by using \(\Delta=\frac{2\pi}{2^q}\), reduces to
\begin{equation}\label{32}
\small
\left|\mu_n-\mu_n^{(L)}\right|
\leq
\frac{2}{K+1}
\sum_{\ell=L+1}^{\infty}
\frac{I_\ell(\kappa)}{I_0(\kappa)}.
\end{equation}
In addition, by using the series representation of the modified Bessel function of the first kind, we have
\begin{equation}
\small
I_\ell(\kappa)
=
\sum_{r=0}^{\infty}
\frac{1}{r!\,\Gamma(r+\ell+1)}
\left(\frac{\kappa}{2}\right)^{2r+\ell},
\end{equation}
and since \(\Gamma(r+\ell+1)\geq \ell!\,r!\), it follows that
\begin{equation}\label{34}
\small
\begin{split}
I_\ell(\kappa)
&\leq
\frac{1}{\ell!}\left(\frac{\kappa}{2}\right)^\ell
\sum_{r=0}^{\infty}
\frac{1}{(r!)^2}
\left(\frac{\kappa}{2}\right)^{2r}.
\end{split}
\end{equation}
Finally, considering that $I_0(\kappa)
=
\sum_{r=0}^{\infty}
\frac{1}{(r!)^2}
\left(\frac{\kappa}{2}\right)^{2r}$, \eqref{34} can be rewritten as
\begin{equation}
\small
\frac{I_\ell(\kappa)}{I_0(\kappa)}
\leq
\frac{1}{\ell!}\left(\frac{\kappa}{2}\right)^\ell.
\end{equation}
Thus, by substituting the above inequality into \eqref{32} and recognizing that
\begin{equation}
\small
\sum_{\ell=L+1}^{\infty}
\frac{1}{\ell!}\left(\frac{\kappa}{2}\right)^\ell
=
\exp\!\left(\frac{\kappa}{2}\right)
-
\sum_{\ell=0}^{L}
\frac{1}{\ell!}\left(\frac{\kappa}{2}\right)^\ell,
\end{equation}
we obtain \eqref{eq:trunc_bound_exp}, which concludes the proof.
\end{IEEEproof}

Moreover, the following corollary specializes Proposition~2 to the Rayleigh (non-LoS) case, corresponding to \(K=0\) and \(\kappa=0\), for which the residual phase error becomes uniformly distributed over the quantization interval.
\begin{corollary}\label{corol2}
For the special case \(K=0\) and \(\kappa=0\), the residual phase error \(\varepsilon\) is uniformly distributed over \([-\Delta/2,\Delta/2]\), and its \(n\)-th circular moment reduces to $\mu_n = \sin\!\left(\frac{n\Delta}{2}\right) / \left(\frac{n\Delta}{2} \right)$.
\end{corollary}
\begin{IEEEproof}
For \(K=0\) and \(\kappa=0\), the PDF in \eqref{error_distribution} reduces to
\begin{equation}
\small
f_{\varepsilon}(\varepsilon)=\frac{1}{\Delta},\qquad \varepsilon \in\left[-\frac{\Delta}{2},\; \frac{\Delta}{2} \right].
\end{equation}
Hence, the $n$-th circular moment can be mathematically expressed as
\begin{equation}
\small
\mu_n=\frac{1}{\Delta}\int_{-\Delta/2}^{\Delta/2} e^{jn\varepsilon}\,d\varepsilon=\frac{1}{\Delta}\frac{e^{jn\Delta/2}-e^{-jn\Delta/2}}{jn}.
\end{equation}
Finally, by using the identity \(e^{jx}-e^{-jx}=2j\sin(x)\), we obtain $\mu_n=\frac{\sin\!\left(\frac{n\Delta}{2}\right)}{\frac{n\Delta}{2}}$, which completes the proof.
\end{IEEEproof}
\begin{remark}
In contrast to~\cite{badiu2020communication}, where the circular moments presented in Corollary \ref{corol2} are considered applicable independently of the channel's phase distribution, they are strictly valid only when the channel's phase is uniformly distributed, which is consistent only with pure non-LoS propagation conditions. As a result, they cannot accurately capture scenarios with fading channels that include LoS components.
\end{remark}

\section{Performance Analysis}

Building on the statistical characterization of the residual phase error, we now proceed to evaluate the performance of the considered zeRIS-assisted network in terms of the joint energy-data rate outage probability. In particular, the derived circular moments provide a tractable way to capture the effect of phase quantization, thereby enabling the analytical treatment of the outage performance under the considered HaR schemes. Finally, since the structure of the cascaded channel depends on the zeRIS placement, the analysis is first presented for the Tx-side deployment and subsequently extended to the UE-side case.

\subsection{Tx-side Deployment}

In the Tx-side deployment, the zeRIS is placed in close proximity to the transmitter, so that the Tx-zeRIS link is considered as a pure LoS link. As a result, the corresponding channel exhibits deterministic behavior, with $|h_{1i}|=1$ and $\arg(h_{1i})=\frac{2\pi d_1}{\lambda}$. In contrast, the zeRIS-UE link is subject to small-scale fading whose amplitude $|h_{2i}|$ follows a Nakagami-$m$ distribution, while its phase is random and incorporated into the residual phase error model, where in this case, the distance $d_u=d_2$. Therefore, under the Tx-side deployment, the randomness of the cascaded channel is entirely characterized by the zeRIS-UE link, while the Tx-zeRIS link contributes only a deterministic phase shift. In this direction, we derive a closed-form approximation for the joint energy–data rate outage probability of a zeRIS with quantized phase shifts under the Tx-side deployment with the TS scheme.

\begin{proposition}\label{prop:4}
Under the Tx-side deployment and the TS scheme, the joint energy–data rate outage probability of a zeRIS-assisted network can be accurately approximated as
\begin{equation}\label{PBTS}
\small
P_{\mathrm{B}}^{\mathrm{TS}}
=
\begin{cases}
1, & \text{if } N \leq N_{\min}^{\mathrm{TS}}, \\
\frac{\gamma\!\left(k_{\mathrm{TS}}, \frac{x^B_{\mathrm{TS}}}{\theta_{\mathrm{TS}}}\right)}{\Gamma\!\left(k_{\mathrm{TS}}\right)}, & \text{otherwise},
\end{cases}
\end{equation}
where $x^B_{\mathrm{TS}} = \frac{2^{\frac{R_{\mathrm{thr}}}{1-\tau}} - 1}{\gamma_t G \ell_p}$, \(\Gamma(\cdot)\) denotes the Gamma function, \(\gamma(\cdot,\cdot)\) denotes the lower incomplete Gamma function, $N_{\min}^{\mathrm{TS}}=
\frac{(1-\tau)P_{\mathrm{elem}}+\sqrt{(1-\tau)^2P_{\mathrm{elem}}^2+4\tau\zeta P_t G_t \ell_1 P_{\mathrm{ctrl}}}}{2\tau\zeta P_t G_t \ell_1}$, $k_{\mathrm{TS}}=
\frac{\mathbb{E}^2[X;N,d_2]}{\mathbb{E}[X^2;N,d_2]-\mathbb{E}^2[X;N,d_2]}$, and $\theta_{\mathrm{TS}}=
\frac{\mathbb{E}[X^2;N,d_2]-\mathbb{E}^2[X;N,d_2]}{\mathbb{E}[X;N,d_2]}$, where $\mathbb{E}[X;N,d_2]$ and $\mathbb{E}[X^2;N,d_2]$ are given in \eqref{eq:EX_TS} and \eqref{eq:EX2_TS}, respectively.
\end{proposition}
\begin{figure*}[t!]

\begin{equation}\label{eq:EX_TS}
\small
\mathbb{E}[X;N,d_u]
=
N\frac{\Gamma(m_n+1)}{\Gamma(m_n)}\left(\frac{\Omega}{m_n}\right)
+
N(N-1)
\left(
\frac{\Gamma\!\left(m_n+\frac{1}{2}\right)}{\Gamma(m_n)}
\left(\frac{\Omega}{m_n}\right)^{\frac{1}{2}}
\right)^2
|\mu_1(d_u)|^2
\end{equation}
\hrule

\begin{equation}\label{eq:EX2_TS}
\small
\begin{split}
&\mathbb{E}[X^2;N,d_u]
=
N\frac{\Gamma(m_n+2)}{\Gamma(m_n)}\left(\frac{\Omega}{m_n}\right)^2 +
4N(N-1)
\frac{\Gamma(m_n+1)}{\Gamma(m_n)}\left(\frac{\Omega}{m_n}\right)
\left(
\frac{\Gamma\!\left(m_n+\frac{1}{2}\right)}{\Gamma(m_n)}
\left(\frac{\Omega}{m_n}\right)^{\frac{1}{2}}
\right)^2
|\mu_1(d_u)|^2 \\
&\quad +
N(N-1)
\left(
\frac{\Gamma(m_n+1)}{\Gamma(m_n)}\left(\frac{\Omega}{m_n}\right)
\right)^2
|\mu_2(d_u)|^2 +
2N(N-1)
\left(
\frac{\Gamma(m_n+1)}{\Gamma(m_n)}\left(\frac{\Omega}{m_n}\right)
\right)^2 \\
&\quad +
N(N-1)(N-2)
\frac{\Gamma(m_n+1)}{\Gamma(m_n)}\left(\frac{\Omega}{m_n}\right)
\left(
\frac{\Gamma\!\left(m_n+\frac{1}{2}\right)}{\Gamma(m_n)}
\left(\frac{\Omega}{m_n}\right)^{\frac{1}{2}}
\right)^2  
\Big(2\mathrm{Re}\!\{\mu_2(d_u)(\mu_1^{*}(d_u))^2\}+4|\mu_1(d_u)|^2\Big) \\
&\quad +
N(N-1)(N-2)(N-3)
\left(
\frac{\Gamma\!\left(m_n+\frac{1}{2}\right)}{\Gamma(m_n)}
\left(\frac{\Omega}{m_n}\right)^{\frac{1}{2}}
\right)^4
|\mu_1(d_u)|^4
\end{split}
\end{equation}

\hrule
\end{figure*}
\begin{IEEEproof}
 The proof can be found in Appendix A.
\end{IEEEproof}
\begin{remark}
Since for \(N > N_{\min}^{\mathrm{TS}}\) the joint energy–data rate outage probability reduces to the rate outage probability and \(x^B_{\mathrm{TS}}\) increases in \(\tau\), the outage probability monotonically increases with \(\tau\), and thus the optimal time-switching factor $\tau_{\mathrm{B}}^{*}$ is obtained at the boundary of the energy feasibility constraint, yielding
\begin{equation}
\small
\tau_{\mathrm{B}}^{*} = \frac{N P_{\mathrm{elem}} + P_{\mathrm{ctrl}}}{N P_{\mathrm{elem}} + N^2 \zeta P_t G_t \ell_1}.
\end{equation}
\end{remark}

Next, we derive the joint energy–data rate outage probability for the case where a Tx-side zeRIS harvests energy through the ES scheme.
\begin{proposition}
The joint energy–data rate outage probability for a Tx-side zeRIS that applies the ES scheme can be accurately approximated as
\begin{equation}\label{PBES}
\small
P^{\mathrm{ES}}_{\mathrm{B}} =
\begin{cases}
1, & N_1 \leq N_{\min}^{\mathrm{ES}}, \\
\dfrac{\gamma\!\left(k_{\mathrm{ES}}, \frac{x^B_{\mathrm{ES}}}{\theta_{\mathrm{ES}}}\right)}{\Gamma\!\left(k_{\mathrm{ES}}\right)}, & \text{otherwise},
\end{cases}
\end{equation}
where $x^B_{\mathrm{ES}}=\frac{2^{R_{\mathrm{thr}}}-1}{\gamma_t G \ell_1 \ell_2}$, $N_{\min}^{\mathrm{ES}}=\sqrt{\frac{N_2 P_{\mathrm{elem}}+P_{\mathrm{ctrl}}}{\zeta P_t G_t \ell_1}}$, $k_{\mathrm{ES}}= \frac{\mathbb{E}_{[X]}^2(N_2,d_2)}{\mathbb{E}_{[X^2]}(N_2,d_2)-\mathbb{E}_{[X]}^2(N_2,d_2)}$, and $\theta_{\mathrm{ES}}=
\frac{\mathbb{E}_{[X^2]}(N_2,d_2)-\mathbb{E}_{[X]}^2(N_2,d_2)}{\mathbb{E}_{[X]}(N_2,d_2)}$.
\end{proposition}

\begin{IEEEproof}
Considering the definition of the joint energy-data rate outage probability, it follows that for the Tx-side deployment and the ES scheme it can be expressed as
\begin{equation}\label{PGEN_ES_BS}
\small
\begin{split}
&P_{\mathrm{B}}^{\mathrm{ES}}
= \Pr \Bigg( T \zeta P_t G_t \ell_1
\left|\sum_{i=1}^{N_1} e^{j\left(\omega_i^{\mathrm{h}}+\arg(h_{1i})\right)}\right|^2 \\
&\leq T\big(N_2P_{\mathrm{elem}}+P_{\mathrm{ctrl}}\big) \\
&\cup \ \log_2\!\left(1+\gamma_t G \ell_1 \ell_2
\left|\sum_{i=1}^{N_2}|h_{2i}|e^{j\varepsilon_i}\right|^2\right) \leq R_{\mathrm{thr}} \Bigg).
\end{split}
\end{equation}
Similarly to the TS scheme, for the Tx-side deployment the Tx-zeRIS link is purely LoS, and thus \(\arg(h_{1i})=\frac{2\pi d_1}{\lambda}\) for all zeRIS elements. As a result, during the energy harvesting phase, all \(N_1\) harvesting elements apply the same quantized phase shift, which leads to a common residual phase and preserves full coherence. Therefore, quantization does not affect the absorption process, and the magnitude $\left|\sum_{i=1}^{N_1} e^{j\left(\omega_i^{\mathrm{h}}+\arg(h_{1i})\right)}\right|^2$ becomes equal to $N_1^2$. Thus, the harvested energy becomes independent of the residual phase error, and the energy outage event can be equivalently written as  \(P_{\mathrm{B}}^{\mathrm{ES}}=1\) when \(N_1 \leq N_{\min}^{\mathrm{ES}}\).

For the case \(N_1 > N_{\min}^{\mathrm{ES}}\), the outage event is determined only by the rate condition. In particular, by defining $X=\left|\sum_{i=1}^{N_2}|h_{2i}|e^{j\varepsilon_i}\right|^2$, the rate outage event can be rewritten as 
\begin{equation}
\small
P_{\mathrm{B}}^{\mathrm{ES}}=\Pr\!\left(X \leq x^B_{\mathrm{ES}}\right).
\end{equation}
Then, by following the same moment matching procedure as in the TS scheme and using the moments \(\mathbb{E}_{[X]}(N_2,d_2)\) and \(\mathbb{E}_{[X^2]}(N_2,d_2)\), the parameters \(k_{\mathrm{ES}}\) and \(\theta_{\mathrm{ES}}\) can be obtained, which directly yields \eqref{PBES} and concludes the proof.
\end{IEEEproof}
\begin{remark}
Since for \(N_1 > N_{\min}^{\mathrm{ES}}\) the outage probability is dictated by the rate term which decreases with \(N_1\), the optimal partition is obtained at the boundary of the energy feasibility constraint, yielding
\begin{equation}
\small
N_{1,\mathrm{B}}^{*} = \frac{-P_{\mathrm{elem}} + \sqrt{P_{\mathrm{elem}}^2 + 4\zeta P_t G_t \ell_1 \left( N P_{\mathrm{elem}} + P_{\mathrm{ctrl}}\right)}}{2 \zeta P_t G_t \ell_1}.
\end{equation}
\end{remark}
\begin{remark}
The derived expressions for the joint energy-data rate outage probability can be directly utilized to characterize the energy efficiency of the zeRIS-assisted system, as they capture the successful transmission behavior under explicit energy constraints, in line with the definition adopted in~\cite{zeris}.
\end{remark}

\subsection{UE-side Deployment}

In the UE-side deployment, the zeRIS is placed in close proximity to the UE so that the zeRIS-UE link is considered as a pure LoS link. As a result, the corresponding channel exhibits deterministic behavior, with $|h_{2i}|=1$ and $\arg(h_{2i})=\frac{2\pi d_2}{\lambda}$. In contrast, the Tx-zeRIS link is subject to small-scale fading whose amplitude $|h_{1i}|$ follows a Nakagami-$m$ distribution, while its phase is random and incorporated into the residual phase error model, where in this case $d_u=d_1$. Therefore, under the UE-side deployment, the randomness of the cascaded channel is entirely characterized by the Tx-zeRIS link, while the zeRIS-UE link contributes only a deterministic phase shift. In this direction, we derive a closed-form expression for the joint energy–data rate outage probability of a zeRIS with quantized phase shifts under the UE-side deployment with the TS scheme.
\begin{proposition}
Under the UE-side deployment and the TS scheme, the joint energy–data rate outage probability of a zeRIS-assisted network can be accurately approximated as
\begin{equation}\label{PUTS}
\small
P_{\mathrm{U}}^{\mathrm{TS}}
=
\frac{\gamma\!\left(k_{\mathrm{TS}}, \frac{x^U_{\mathrm{TS}}}{\theta_{\mathrm{TS}}}\right)}{\Gamma\!\left(k_{\mathrm{TS}}\right)},
\end{equation}
where $x^U_{\mathrm{TS}}=\max \left\{
\frac{2^{\frac{R_{\mathrm{thr}}}{1-\tau}} - 1}{\gamma_t G \ell_1 \ell_2},
\frac{N P_{\mathrm{elem}}(1-\tau) + P_{\mathrm{ctrl}}}{\tau \zeta P_t G_t \ell_1}
\right\}$, $k_{\mathrm{TS}}= \frac{\mathbb{E}^2[X;N,d_1]}{\mathbb{E}[X^2;N,d_1]-\mathbb{E}^2[X;N,d_1]}$, and $\theta_{\mathrm{TS}}= \frac{\mathbb{E}[X^2;N,d_1]-\mathbb{E}^2[X;N,d_1]}{\mathbb{E}[X;N,d_1]}$.
\end{proposition}

\begin{IEEEproof}
Considering the definition of the joint energy-data rate outage probability, it follows that for the UE-side deployment and the TS scheme it can be expressed as
\begin{equation}\label{PGEN_TS_UE}
\small
\begin{split}
&P_{\mathrm{U}}^{\mathrm{TS}} 
=\Pr \!\Bigg( \tau T \zeta P_t G_t \ell_1 
\left|\sum_{i=1}^{N} |h_{1i}| e^{j\left(\omega_i^{\mathrm{h}}+\arg(h_{1i})\right)}\right|^2 \\
&\leq T\big((1-\tau)NP_{\mathrm{elem}}+P_{\mathrm{ctrl}}\big) \cup \!(1-\tau)\\& \times\log_2\!\left(\!1\!+\!\gamma_t G \ell_p 
\left|\sum_{i=1}^{N}|h_{1i}| e^{j\left(\omega_i^{\mathrm{c}}+\arg(h_{1i})+\arg(h_{2i})\right)}\right|^2\right) \!\!\leq \!R_{\mathrm{thr}}\!\! \Bigg).
\end{split}
\end{equation}
Since the zeRIS-UE link is purely LoS, \(\arg(h_{2i})=\frac{2\pi d_2}{\lambda}\) is deterministic for all elements, and thus the corresponding residual phase errors have identical statistics, which implies that both the harvested energy and rate terms can be expressed through the same residual phase error RV \(\varepsilon_i\), yielding
\begin{equation}
\small
\begin{split}
P_{\mathrm{U}}^{\mathrm{TS}}
= \Pr \Bigg(
&\left|\sum_{i=1}^{N}|h_{1i}|e^{j\varepsilon_i}\right|^2
\leq
\frac{N P_{\mathrm{elem}}(1-\tau) + P_{\mathrm{ctrl}}}{\tau \zeta P_t G_t \ell_1}
\\
&\cup
\left|\sum_{i=1}^{N}|h_{1i}|e^{j\varepsilon_i}\right|^2
\leq
\frac{2^{\frac{R_{\mathrm{thr}}}{1-\tau}} - 1}{\gamma_t G \ell_p }
\Bigg),
\end{split}
\end{equation}
and since both events involve the same RV, their union reduces to a single inequality of the form \(P_{\mathrm{U}}^{\mathrm{TS}} = \Pr\!\left(\left|\sum_{i=1}^{N}|h_{1i}|e^{j\varepsilon_i}\right|^2 \leq x^U_{\mathrm{TS}}\right)\). Therefore, by approximating the RV \(X=\left|\sum_{i=1}^{N}|h_{1i}|e^{j\varepsilon_i}\right|^2\) via moment matching with a Gamma distribution using the moments \(\mathbb{E}[X;N,d_1]\) and \(\mathbb{E}[X^2;N,d_1]\), we directly obtain \eqref{PUTS}, which concludes the proof.
\end{IEEEproof}
\begin{remark}
Since the outage probability is determined by the maximum of the energy and rate thresholds, the optimal time-switching factor $\tau_{\mathrm{U}}^{*}$ is obtained by equating the two terms in $x^U_{\mathrm{TS}}$.
\end{remark}

Finally, we proceed to derive the joint energy–data rate outage probability for the case where a UE-side zeRIS harvests energy through the ES scheme.
\begin{proposition}
The joint energy–data rate outage probability for a UE-side zeRIS that applies the ES scheme can be accurately approximated as
\begin{equation}\label{PUES}
\small
P^{\mathrm{ES}}_{\mathrm{U}}
=
\frac{\gamma\!\left(k_1,\frac{x^U_{\mathrm{ES}}}{\theta_1}\right)}{\Gamma(k_1)}
+
\frac{\gamma\!\left(k_2,\frac{x^U_{\mathrm{R}}}{\theta_2}\right)}{\Gamma(k_2)}
-
\frac{\gamma\!\left(k_1,\frac{x^U_{\mathrm{ES}}}{\theta_1}\right)}{\Gamma(k_1)}
\frac{\gamma\!\left(k_2,\frac{x^U_{\mathrm{R}}}{\theta_2}\right)}{\Gamma(k_2)},
\end{equation}
where $x^U_{\mathrm{ES}}=\frac{N_2 P_{\mathrm{elem}}+P_{\mathrm{ctrl}}}{\zeta P_t G_t \ell_1}$, $x^U_{\mathrm{R}}=\frac{2^{R_{\mathrm{thr}}}-1}{\gamma_t G \ell_p}$, $k_1=\frac{\mathbb{E}_{[X]}^2(N_1,d_1)}{\mathbb{E}_{[X^2]}(N_1,d_1)-\mathbb{E}_{[X]}^2(N_1,d_1)}$, $\theta_1=\frac{\mathbb{E}_{[X^2]}(N_1,d_1)-\mathbb{E}_{[X]}^2(N_1,d_1)}{\mathbb{E}_{[X]}(N_1,d_1)}$, $k_2=\frac{\mathbb{E}_{[X]}^2(N_2,d_1)}{\mathbb{E}_{[X^2]}(N_2,d_1)-\mathbb{E}_{[X]}^2(N_2,d_1)}$, and $\theta_2=\frac{\mathbb{E}_{[X^2]}(N_2,d_1)-\mathbb{E}_{[X]}^2(N_2,d_1)}{\mathbb{E}_{[X]}(N_2,d_1)}$.
\end{proposition}
\begin{IEEEproof}
Considering the definition of the joint energy-data rate outage probability, it follows that for the UE-side deployment and the ES scheme it can be expressed as
\begin{equation}\label{PGEN_ES_UE}
\small
\begin{split}
&P_{\mathrm{U}}^{\mathrm{ES}}
= \Pr \Bigg(
T \zeta P_t G_t \ell_1
\left|\sum_{i=1}^{N_1}|h_{1i}|e^{j\left(\omega_i^{\mathrm{h}}+\arg(h_{1i})\right)}\right|^2
\\
&\leq T\left(
N_2P_{\mathrm{elem}}+P_{\mathrm{ctrl}}
\right)
\cup\,
\log_2\!\Bigg(
1+\gamma_t G \ell_p
\left|\sum_{i=1}^{N_2}|h_{1i}|
\right.
\\
&\left.
\qquad \times e^{j\left(\omega_i^{\mathrm{c}}+\arg(h_{1i})+\arg(h_{2i})\right)}\right|^2
\Bigg)
\leq R_{\mathrm{thr}}
\Bigg).
\end{split}
\end{equation}
Since the zeRIS-UE link is purely LoS, \(\arg(h_{2i})=\frac{2\pi d_2}{\lambda}\) is deterministic for all elements, and thus the corresponding residual phase errors in the harvesting and reflection terms have identical statistics. Therefore, by denoting these residual phase errors by \(\varepsilon_i\), \eqref{PGEN_ES_UE} can be rewritten as
\begin{equation}
\small
\begin{split}
P_{\mathrm{U}}^{\mathrm{ES}}
=\Pr \Bigg(
&\left|\sum_{i=1}^{N_1}|h_{1i}|e^{j\varepsilon_i}\right|^2
\!\!\!\leq \!x^U_{\mathrm{ES}}\cup
\left|\sum_{i=1}^{N_2}|h_{1i}|e^{j\varepsilon_i}\right|^2
\leq x^U_{\mathrm{R}}
\Bigg).
\end{split}
\end{equation}
Since the two events correspond to different element subsets of sizes \(N_1\) and \(N_2\), they are independent, and thus the outage probability becomes
\begin{equation}
\small
\begin{split}
P_{\mathrm{U}}^{\mathrm{ES}}
& \!=\!
\Pr\!\left(
\left|\sum_{i=1}^{N_1}|h_{1i}|e^{j\varepsilon_i}\right|^2
\! \leq \! x^U_{\mathrm{ES}}
\right)
+
\Pr\!\left(
\left|\sum_{i=1}^{N_2}|h_{1i}|e^{j\varepsilon_i}\right|^2
\!\leq \! x^U_{\mathrm{R}}
\right)
\\
& -
\Pr\!\left(
\left|\sum_{i=1}^{N_1}|h_{1i}|e^{j\varepsilon_i}\right|^2
\! \leq \! x^U_{\mathrm{ES}}
\right)
\Pr\!\left(
\left|\sum_{i=1}^{N_2}|h_{1i}|e^{j\varepsilon_i}\right|^2
\! \leq \! x^U_{\mathrm{R}}
\right) \! .
\end{split}
\end{equation}
Finally, by approximating the two RVs through Gamma distributions using moment matching with the moments \(\mathbb{E}_{[X]}(N_1,d_1)\), \(\mathbb{E}_{[X^2]}(N_1,d_1)\), \(\mathbb{E}_{[X]}(N_2,d_1)\), and \(\mathbb{E}_{[X^2]}(N_2,d_1)\), we obtain \eqref{PUES}, which concludes the proof.
\end{IEEEproof}

\begin{table}[!ht]
\footnotesize
\renewcommand{\arraystretch}{0.95}
\caption{\textsc{Simulation Parameters}}
\label{values_sim}
\centering
\begin{tabular}{lll}
\hline
\bfseries Parameter & \bfseries Notation & \bfseries Value \\
\hline\hline
Carrier frequency & $f$ & $900$ MHz \\
Block duration & $T$ & $1$ s \\
Path loss constant & $C_0$ & $\lambda^2/(16\pi^2)$ \\
Energy efficiency & $\zeta$ & $0.65$ \\
Element power & $P_{\mathrm{elem}}$ & $q \times 0.06$ mW \\
Controller power & $P_{\mathrm{circ}}$ & $50$ mW \\
Noise variance & $\sigma^2$ & $-100$ dB \\
Tx antenna gain & $G_t$ & $4$ dB \\
UE antenna gain & $G_r$ & $0$ dB \\
Rate threshold & $R_{\mathrm{thr}}$ & $\log_2(11)$ \\
Nakagami shape & $m_n$ & $3$ \\
Nakagami spread & $\Omega$ & $1$ \\
Concentration parameter & $\kappa$ & $3$ \\
\hline
\end{tabular}
\end{table}

\section{Numerical Results}

In this section, we evaluate the performance of a downlink zeRIS-assisted communication scenario in terms of joint energy-data rate outage probability, minimum number of reflecting elements, and energy efficiency. Specifically, both the Tx-side and UE-side zeRIS cases are investigated, while the impact of the considered HaR schemes and phase quantization is examined. To derive the numerical results, we set the parameters of the analyzed system model as shown in Table \ref{values_sim}, where it should be noted that $\kappa$ is set to be 3 in order to describe a wireless propagation environment with non-isotropically distributed scatterers that impact the phase of the received signal \cite{Abdi}. Furthermore, unless otherwise stated, the distances are set to $d_1=15$ m and $d_2=45$ m for the Tx-side deployment, and $d_1=45$ m and $d_2=15$ m for the UE-side deployment, where the link in close proximity to the zeRIS is modeled as a LoS link with path loss exponent equal to $2$, whereas the other link is subject to Nakagami-$m$ fading with path loss exponent equal to $2.2$. Additionally, for the numerical evaluation of the circular moments, the infinite series in \eqref{eq:nth_moment} is truncated to $L=10$, which, according to Proposition~3, guarantees a truncation error below $10^{-6}$ for the considered values of $\kappa$, thereby ensuring negligible impact on the reported results. Finally, we employ Monte Carlo simulations with $10^7$ realizations to verify the accuracy of the derived analytical results, where the simulation results are illustrated as marks, whereas the analytical results corresponding to the proposed model and to the benchmark in \cite{badiu2020communication} are illustrated as solid and dashed lines, respectively.

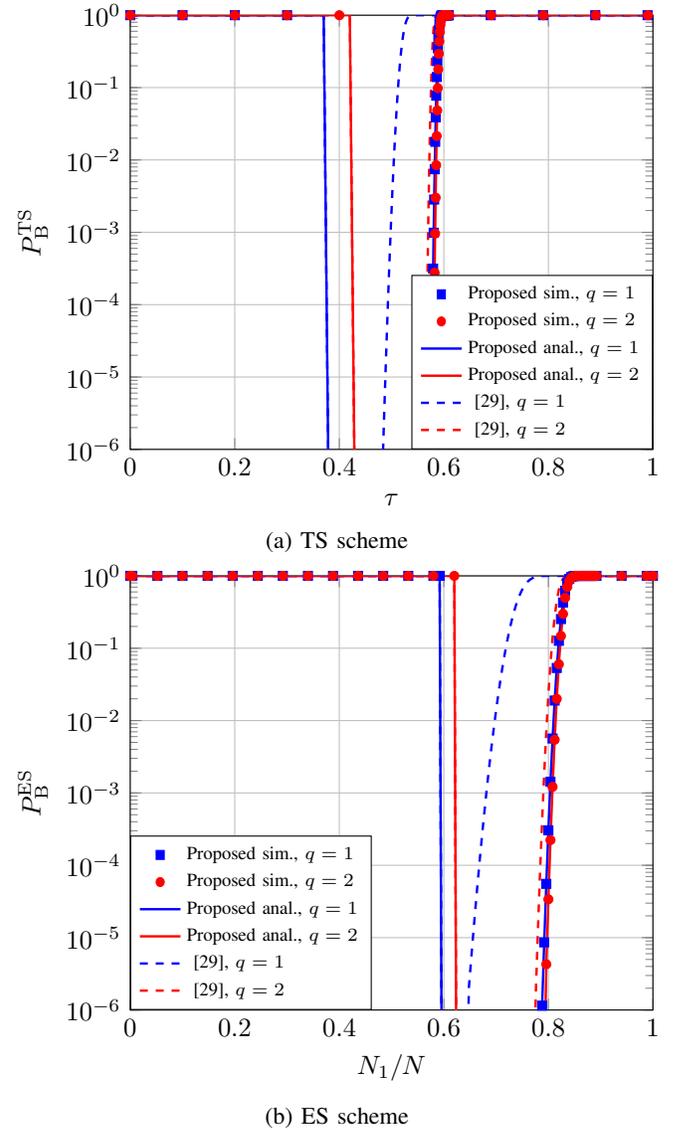
\begin{figure}
    \centering

    \begin{subfigure}[t]{.48\textwidth}
        \centering
        \begin{tikzpicture}
        \begin{semilogyaxis}[
            width=0.98\linewidth,
            xlabel = {$\tau$},
            ylabel = {${P}_{\mathrm{B}}^{\mathrm{TS}}$},
            xmin = 0, xmax = 1,
            ymin = 1e-6, ymax = 1,
            grid = major,
            legend image post style={xscale=0.9},
            legend cell align = {left},
            legend style={at={(1,0)},anchor=south east, font=\scriptsize}
        ]

        \addplot[
            blue,
            only marks,
            mark=square*,
            mark options={solid},
            mark repeat = 10,
            mark size = 1.7,
        ]
        table {results/BSside_TS_q1_proposed_theory.dat};
        \addlegendentry{Proposed sim., $q=1$}

        \addplot[
            red,
            only marks,
            mark=*,
            mark options={solid},
            mark repeat = 10,
            mark size = 1.7,
        ]
        table {results/BSside_TS_q2_proposed_theory.dat};
        \addlegendentry{Proposed sim., $q=2$}

        \addplot[
            blue,
            no marks,
            line width = 0.9pt,
            solid,
        ]
        table {results/BSside_TS_q1_proposed_theory.dat};
        \addlegendentry{Proposed anal., $q=1$}

        \addplot[
            red,
            no marks,
            line width = 0.9pt,
            solid,
        ]
        table {results/BSside_TS_q2_proposed_theory.dat};
        \addlegendentry{Proposed anal., $q=2$}

        \addplot[
            blue,
            no marks,
            line width = 0.9pt,
            dashed,
        ]
        table {results/BSside_TS_q1_badiu_theory.dat};
        \addlegendentry{\cite{badiu2020communication}, $q=1$}

        \addplot[
            red,
            no marks,
            line width = 0.9pt,
            dashed,
        ]
        table {results/BSside_TS_q2_badiu_theory.dat};
        \addlegendentry{\cite{badiu2020communication}, $q=2$}

        \end{semilogyaxis}
        \end{tikzpicture}
        \caption{TS scheme}
        \label{fig:BSside_outage_TS}
    \end{subfigure}
    \hfill
\begin{subfigure}[t]{.48\textwidth}
    \centering
    \begin{tikzpicture}
    \begin{semilogyaxis}[
        width=0.98\linewidth,
        xlabel = {$N_1/N$},
        ylabel = {${P}_{\mathrm{B}}^{\mathrm{ES}}$},
        xmin = 0, xmax = 1,
        ymin = 1e-6, ymax = 1,
        grid = major,
        legend image post style={xscale=0.9},
        legend cell align = {left},
        legend style={at={(0,0)},anchor=south west, font=\scriptsize}
    ]

    \addplot[
        blue,
        only marks,
        mark=square*,
        mark options={solid},
        mark repeat = 1,
        mark size = 1.7,
    ]
    table {results/BSside_ES_q1_proposed_theory.dat};
    \addlegendentry{Proposed sim., $q=1$}

    \addplot[
        red,
        only marks,
        mark=*,
        mark options={solid},
        mark repeat = 1,
        mark size = 1.7,
    ]
    table {results/BSside_ES_q2_proposed_theory.dat};
    \addlegendentry{Proposed sim., $q=2$}

    \addplot[
        blue,
        no marks,
        line width = 0.9pt,
        solid,
    ]
    table {results/BSside_ES_q1_proposed_theory.dat};
    \addlegendentry{Proposed anal., $q=1$}

    \addplot[
        red,
        no marks,
        line width = 0.9pt,
        solid,
    ]
    table {results/BSside_ES_q2_proposed_theory.dat};
    \addlegendentry{Proposed anal., $q=2$}

    \addplot[
        blue,
        no marks,
        line width = 0.9pt,
        dashed,
    ]
    table {results/BSside_ES_q1_badiu_theory.dat};
    \addlegendentry{\cite{badiu2020communication}, $q=1$}

    \addplot[
        red,
        no marks,
        line width = 0.9pt,
        dashed,
    ]
    table {results/BSside_ES_q2_badiu_theory.dat};
    \addlegendentry{\cite{badiu2020communication}, $q=2$}

    \end{semilogyaxis}
    \end{tikzpicture}
    \caption{ES scheme}
    \label{fig:BSside_outage_ES}
\end{subfigure}

    \caption{Joint energy-data rate outage probability for the Tx-side zeRIS under the (a) TS scheme, and (b) ES scheme for $N=250$.}
    \label{fig:BSside_outage}
\end{figure}

Fig. \ref{fig:BSside_outage} illustrates the joint energy-data rate outage probability of a Tx-side zeRIS-assisted network under the TS and ES schemes for $N=250$, where Fig. \ref{fig:BSside_outage}a corresponds to the TS scheme and Fig. \ref{fig:BSside_outage}b corresponds to the ES scheme. The outage probability is depicted as a function of the time switching factor $\tau$ and the element-splitting ratio $N_1/N$, respectively, for different phase resolutions $q$. As can be observed, the analytical results corresponding to the proposed model closely match the simulation results, thereby validating the accuracy of the derived expressions. Furthermore, Fig. \ref{fig:BSside_outage} reveals that increasing the quantization resolution does not necessarily improve the performance, which is attributed to the combined effect of the increased power consumption of the zeRIS elements and the adopted residual phase error model. In particular, higher $q$ increases the power consumption of the reflecting elements, limiting the energy available for the communication phase, while, as captured by the residual phase error distribution in \eqref{error_distribution}, a portion of the impinging phases is concentrated around a deterministic LoS component, represented by the Dirac term weighted by $\frac{K}{K+1}$. In this case, the phases are inherently aligned and quantization effectively results in a phase shift without degrading coherence, thereby limiting the gains from increasing the phase resolution. In contrast, in the absence of a LoS component, increasing $q$ directly improves phase alignment, explaining why \cite{badiu2020communication} tends to overestimate the impact of quantization under the considered conditions. Specifically, as observed in both Fig. \ref{fig:BSside_outage}a and Fig. \ref{fig:BSside_outage}b, the results of \cite{badiu2020communication} exhibit a noticeable deviation from the proposed analysis and consistently operate as a loose lower bound of the achievable performance. Finally, it can be seen that as the quantization resolution increases, the gap between the benchmark and the proposed results decreases, since the range $[\frac{-\Delta}{2}, \frac{\Delta}{2}]$ of the residual phase error becomes smaller and the influence of the underlying distribution diminishes. Therefore, Fig. \ref{fig:BSside_outage} highlights the importance of jointly accounting for hardware consumption and residual phase error statistics in the design of zeRIS-assisted systems.

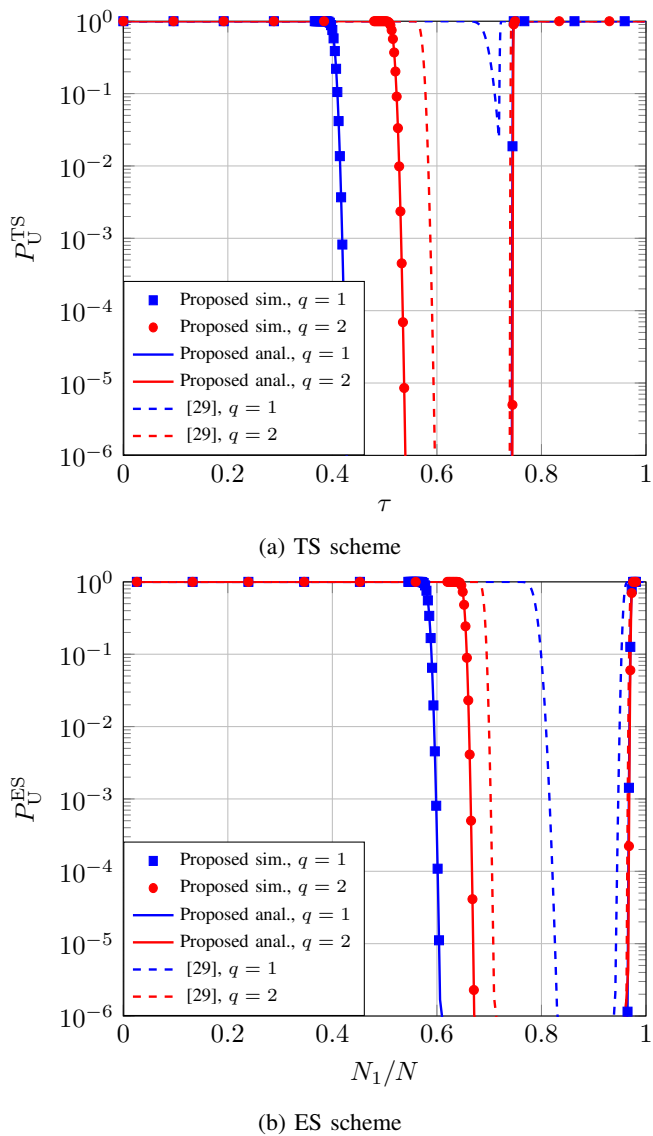
\begin{figure}
    \centering

    \begin{subfigure}[t]{.48\textwidth}
        \centering
        \begin{tikzpicture}
        \begin{semilogyaxis}[
            width=0.98\linewidth,
            xlabel = {$\tau$},
            ylabel = {${P}_{\mathrm{U}}^{\mathrm{TS}}$},
            xmin = 0, xmax = 1,
            ymin = 1e-6, ymax = 1,
            grid = major,
            legend image post style={xscale=0.9},
            legend cell align = {left},
            legend style={at={(0,0)},anchor=south west, font=\scriptsize}
        ]

        \addplot[
            blue,
            only marks,
            mark=square*,
            mark repeat = 24,
            mark size = 1.7,
        ]
        table {results/UE_TS_q1_proposed.dat};
        \addlegendentry{Proposed sim., $q=1$}

        \addplot[
            red,
            only marks,
            mark=*,
            mark repeat = 24,
            mark size = 1.7,
        ]
        table {results/UE_TS_q2_proposed.dat};
        \addlegendentry{Proposed sim., $q=2$}

        \addplot[
            blue,
            no marks,
            line width = 0.9pt,
            solid,
        ]
        table {results/UE_TS_q1_proposed.dat};
        \addlegendentry{Proposed anal., $q=1$}

        \addplot[
            red,
            no marks,
            line width = 0.9pt,
            solid,
        ]
        table {results/UE_TS_q2_proposed.dat};
        \addlegendentry{Proposed anal., $q=2$}

        \addplot[
            blue,
            no marks,
            dashed,
            line width = 0.9pt,
        ]
        table {results/UE_TS_q1_badiu.dat};
        \addlegendentry{\cite{badiu2020communication}, $q=1$}

        \addplot[
            red,
            no marks,
            dashed,
            line width = 0.9pt,
        ]
        table {results/UE_TS_q2_badiu.dat};
        \addlegendentry{\cite{badiu2020communication}, $q=2$}

        \end{semilogyaxis}
        \end{tikzpicture}
        \caption{TS scheme}
    \end{subfigure}
    \hfill
    \begin{subfigure}[t]{.48\textwidth}
        \centering
        \begin{tikzpicture}
        \begin{semilogyaxis}[
            width=0.98\linewidth,
            xlabel = {$N_1/N$},
            ylabel = {${P}_{\mathrm{U}}^{\mathrm{ES}}$},
            xmin = 0, xmax = 1,
            ymin = 1e-6, ymax = 1,
            grid = major,
            legend image post style={xscale=0.9},
            legend cell align = {left},
            legend style={at={(0,0)},anchor=south west, font=\scriptsize}
        ]

        \addplot[
            blue,
            only marks,
            mark=square*,
            mark repeat = 4,
            mark size = 1.7,
        ]
        table {results/UE_ES_q1_proposed_theory.dat};
        \addlegendentry{Proposed sim., $q=1$}

        \addplot[
            red,
            only marks,
            mark=*,
            mark repeat = 4,
            mark size = 1.7,
        ]
        table {results/UE_ES_q2_proposed_theory.dat};
        \addlegendentry{Proposed sim., $q=2$}

        \addplot[
            blue,
            no marks,
            line width = 0.9pt,
            solid,
        ]
        table {results/UE_ES_q1_proposed_theory.dat};
        \addlegendentry{Proposed anal., $q=1$}

        \addplot[
            red,
            no marks,
            line width = 0.9pt,
            solid,
        ]
        table {results/UE_ES_q2_proposed_theory.dat};
        \addlegendentry{Proposed anal., $q=2$}

        \addplot[
            blue,
            no marks,
            dashed,
            line width = 0.9pt,
        ]
        table {results/UE_ES_q1_badiu_theory.dat};
        \addlegendentry{\cite{badiu2020communication}, $q=1$}

        \addplot[
            red,
            no marks,
            dashed,
            line width = 0.9pt,
        ]
        table {results/UE_ES_q2_badiu_theory.dat};
        \addlegendentry{\cite{badiu2020communication}, $q=2$}

        \end{semilogyaxis}
        \end{tikzpicture}
        \caption{ES scheme}
    \end{subfigure}

    \caption{Joint energy-data rate outage probability for the UE-side zeRIS under the (a) TS scheme, and (b) ES scheme for $N=1500$.}
    \label{fig:UEside_outage}
\end{figure}

Fig. \ref{fig:UEside_outage} illustrates the joint energy-data rate outage probability of a UE-side zeRIS-assisted network under the TS and ES schemes for $N=1500$, where Fig. \ref{fig:UEside_outage}a corresponds to the TS scheme and Fig. \ref{fig:UEside_outage}b corresponds to the ES scheme. As can be observed, the analytical results corresponding to the proposed model closely match the simulation results, thereby validating the accuracy of the derived expressions. Moreover, compared to the Tx-side case, a significantly larger number of reflecting elements is required to achieve reliable performance, which is attributed to the fact that, in the UE-side deployment, the link responsible for energy harvesting is subject to both small-scale fading and increased path loss, thus reducing the harvested energy and necessitating larger $N$ values. Furthermore, the importance of accurately modeling the residual phase error becomes even more pronounced in this regime, as, for the TS scheme with $q=1$, the benchmark in \cite{badiu2020communication} suggests that the zeRIS-assisted system cannot achieve reliable communication, with the outage probability remaining above $10^{-2}$ for all $\tau$, whereas the proposed model reveals that reliable operation is in fact achievable over a wide range of $\tau$. This discrepancy indicates that neglecting the correct residual phase error statistics may lead to overly pessimistic conclusions and, consequently, to suboptimal system design choices. Finally, similar to the Tx-side case, it can be observed that as the quantization resolution increases, the gap between the benchmark and the proposed results decreases, since the range $[-\Delta/2, \Delta/2]$ of the residual phase error becomes smaller and the influence of the underlying distribution diminishes. Therefore, Fig. \ref{fig:UEside_outage} further emphasizes that an accurate characterization of phase distortions is essential for the reliable and efficient design of zeRIS-assisted systems, particularly in challenging deployment scenarios such as the UE-side case.

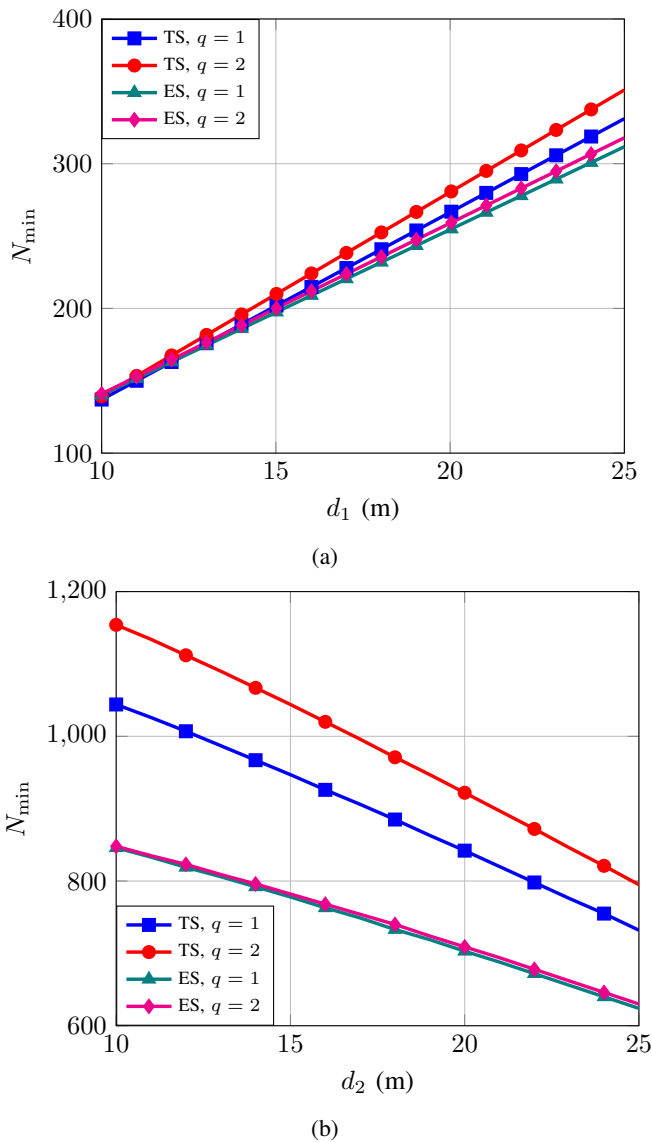
\begin{figure}
    \centering
    \begin{subfigure}[t]{.48\textwidth}
        \centering
        \begin{tikzpicture}
	\begin{axis}[
	width=0.98\linewidth,
	xlabel = {$d_1$ (m)},
	ylabel = {$N_{\mathrm{min}}$},
	xmin = 10,xmax = 25,
	ymin = 100,
	ymax = 400,
	xtick = {10,15,...,25},
	grid = major,
	legend cell align = {left},
        legend style={at={(0,1)},anchor=north west, font=\scriptsize}
	]

	\addplot[
	blue,
        mark=square*,
	mark repeat = 20,
	mark size = 2,
	line width = 1.3pt,
	]
	table {results/BS_TS_Nmin_q1_proposed.dat};
 	\addlegendentry{TS, $q=1$}

	\addplot[
	red,
        mark=*,
	mark repeat = 20,
	mark size = 2,
	line width = 1.3pt,
	]
	table {results/BS_TS_Nmin_q2_proposed.dat};	
 	\addlegendentry{TS, $q=2$}

	\addplot[
	teal,
        mark=triangle*,
	mark repeat = 20,
	mark size = 2,
	line width = 1.3pt,
	]
	table {results/BS_ES_Nmin_q1_proposed.dat};	
 	\addlegendentry{ES, $q=1$}

	\addplot[
	magenta,
        mark=diamond*,
	mark repeat = 20,
	mark size = 2,
	line width = 1.3pt,
	]
	table {results/BS_ES_Nmin_q2_proposed.dat};	
 	\addlegendentry{ES, $q=2$}

	\end{axis}
        \end{tikzpicture}
        \caption{}
        \label{fig:BSNmin}
    \end{subfigure}
    \hfill
    \begin{subfigure}[t]{.48\textwidth}
        \centering
        \begin{tikzpicture}
 	\begin{axis}[
	width=0.98\linewidth,
	xlabel = {$d_2$ (m)},
	ylabel = {$N_{\mathrm{min}}$},
	xmin = 10,xmax = 25,
	ymin = 600,
	ymax = 1200,
	xtick = {10,15,...,25},
	grid = major,
	legend cell align = {left},
        legend style={at={(0,0)},anchor=south west, font=\scriptsize}
	]

	\addplot[
	blue,
        mark=square*,
	mark repeat = 2,
	mark size = 2,
	line width = 1.3pt,
	]
	table {results/UE_TS_Nmin_q1_proposed.dat};
 	\addlegendentry{TS, $q=1$}

	\addplot[
	red,
        mark=*,
	mark repeat = 2,
	mark size = 2,
	line width = 1.3pt,
	]
	table {results/UE_TS_Nmin_q2_proposed.dat};	
 	\addlegendentry{TS, $q=2$}

	\addplot[
	teal,
        mark=triangle*,
	mark repeat = 2,
	mark size = 2,
	line width = 1.3pt,
	]
	table {results/UE_ES_Nmin_q1_proposed.dat};	
 	\addlegendentry{ES, $q=1$}

	\addplot[
	magenta,
        mark=diamond*,
	mark repeat = 2,
	mark size = 2,
	line width = 1.3pt,
	]
	table {results/UE_ES_Nmin_q2_proposed.dat};	
 	\addlegendentry{ES, $q=2$}

	\end{axis}
        \end{tikzpicture}
        \caption{}
        \label{fig:USNmin}
    \end{subfigure}

    \caption{Analysis of minimum reflecting elements for: (a) Tx-side zeRIS with $d_2=60-d_1$, (b) UE-side zeRIS with $d_1=60-d_2$.}
    \label{fig:Nminanalysis}
\end{figure}

Fig. \ref{fig:Nminanalysis} illustrates the minimum required number of reflecting elements, $N_{\mathrm{min}}$, to achieve a joint energy-data rate outage probability lower than or equal to $10^{-6}$ for both Tx-side and UE-side zeRIS deployment scenarios, where Fig. \ref{fig:Nminanalysis}a corresponds to the Tx-side case with $d_2 = 60 - d_1$, and Fig. \ref{fig:Nminanalysis}b corresponds to the UE-side case with $d_1 = 60 - d_2$. In particular, for each configuration, $N_{\mathrm{min}}$ is obtained by evaluating the outage probability at the optimal operating point of each HaR scheme and identifying the minimum number of reflecting elements that satisfies the outage constraint. As can be observed, for the examined propagation scenarios, the ES scheme requires a smaller number of reflecting elements compared to the TS scheme, highlighting its higher efficiency in utilizing the available harvested energy. Moreover, in both deployment scenarios, the case of $q=1$ results in a lower $N_{\mathrm{min}}$ compared to $q=2$, which is attributed to the increased element power consumption at higher $q$ values and the underlying residual phase error characteristics, which limit the performance gains from larger $q$ under the considered propagation conditions. In addition, Fig. \ref{fig:Nminanalysis}a reveals that $N_{\mathrm{min}}$ increases with $d_1$, indicating that as the Tx is farther from the zeRIS, the reduction in harvested energy requires a larger number of reflecting elements. On the other hand, Fig. \ref{fig:Nminanalysis}b shows a decreasing trend of $N_{\mathrm{min}}$ with $d_2$, since reducing the zeRIS–UE distance improves the effective communication link, thereby relaxing the system requirements. Finally, a clear performance gap between the Tx-side and UE-side deployments is observed, with the latter requiring significantly more reflecting elements, which is consistent with the more challenging energy harvesting conditions in the UE-side case.

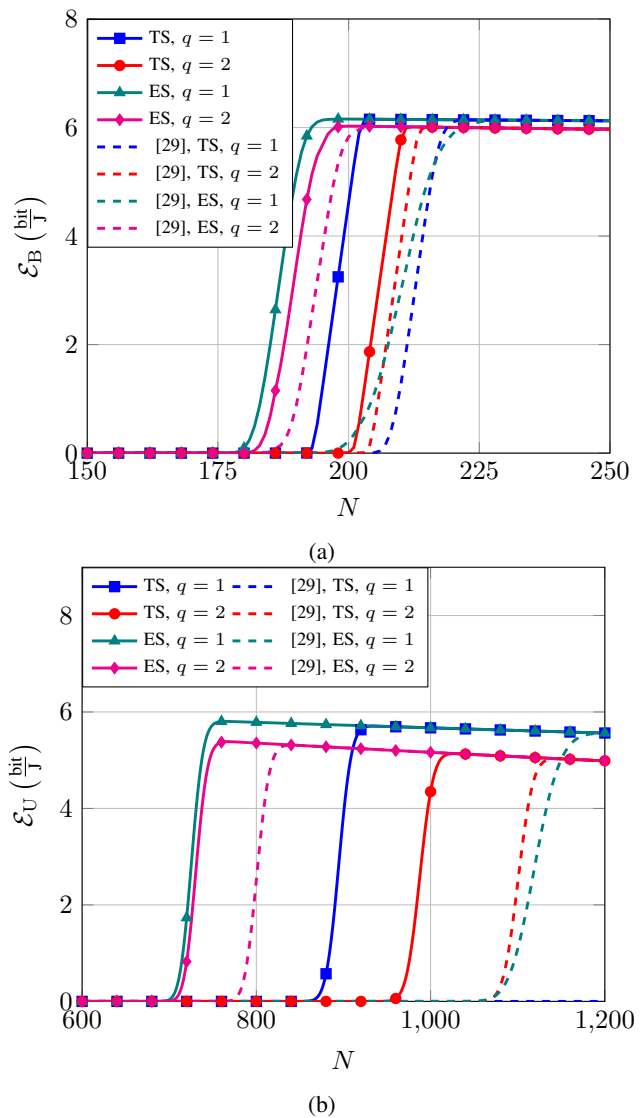
\begin{figure}
    \centering
    \begin{subfigure}[t]{.48\textwidth}
        \centering
        \begin{tikzpicture}
        \begin{axis}[
        width=0.98\linewidth,
        xlabel = {$N$},
        ylabel = {$\mathcal{E}_{\mathrm{B}} \left(\frac{\text{bit}}{\text{J}}\right)$},
        xmin = 150, xmax = 250,
        ymin = 0, ymax = 8,
        xtick = {150,175,200,225,250},
        grid = major,
        legend cell align = {left},
        legend style={at={(0,1)},anchor=north west, font=\scriptsize}
        ]

        \addplot[
        blue, mark=square*, mark repeat=6, mark size=1.6,
        line width=1.1pt
        ] table {results/EE_BS_TS_q1_proposed_smooth.dat};
        \addlegendentry{TS, $q=1$}

        \addplot[
        red, mark=*, mark repeat=6, mark size=1.6,
        line width=1.1pt
        ] table {results/EE_BS_TS_q2_proposed_smooth.dat};
        \addlegendentry{TS, $q=2$}

        \addplot[
        teal, mark=triangle*, mark repeat=6, mark size=1.6,
        line width=1.1pt
        ] table {results/EE_BS_ES_q1_proposed_smooth.dat};
        \addlegendentry{ES, $q=1$}

        \addplot[
        magenta, mark=diamond*, mark repeat=6, mark size=1.6,
        line width=1.1pt
        ] table {results/EE_BS_ES_q2_proposed_smooth.dat};
        \addlegendentry{ES, $q=2$}

        \addplot[blue, dashed, line width=1.1pt]
        table {results/EE_BS_TS_q1_badiu_smooth.dat};
        \addlegendentry{\cite{badiu2020communication}, TS, $q=1$}

        \addplot[red, dashed, line width=1.1pt]
        table {results/EE_BS_TS_q2_badiu_smooth.dat};
        \addlegendentry{\cite{badiu2020communication}, TS, $q=2$}

        \addplot[teal, dashed, line width=1.1pt]
        table {results/EE_BS_ES_q1_badiu_smooth.dat};
        \addlegendentry{\cite{badiu2020communication}, ES, $q=1$}

        \addplot[magenta, dashed, line width=1.1pt]
        table {results/EE_BS_ES_q2_badiu_smooth.dat};
        \addlegendentry{\cite{badiu2020communication}, ES, $q=2$}

        \end{axis}
        \end{tikzpicture}
        \caption{}
        \label{fig:EEBS}
    \end{subfigure}
    \hfill
    \begin{subfigure}[t]{.48\textwidth}
        \centering
        \begin{tikzpicture}
        \begin{axis}[
        width=0.98\linewidth,
        xlabel = {$N$},
        ylabel = {$\mathcal{E}_{\mathrm{U}} \left(\frac{\text{bit}}{\text{J}}\right)$},
        xmin = 600, xmax = 1200,
        ymin = 0, ymax = 9,
        xtick = {600,800,1000,1200},
        grid = major,
        legend columns=4,
        transpose legend,
        legend cell align = {left},
        legend style={at={(0,1)},anchor=north west, font=\scriptsize}
        ]

        \addplot[
        blue, mark=square*, mark repeat=40, mark size=1.6,
        line width=1.1pt
        ] table {results/EE_UE_TS_q1_proposed_smooth.dat};
        \addlegendentry{TS, $q=1$}

        \addplot[
        red, mark=*, mark repeat=40, mark size=1.6,
        line width=1.1pt
        ] table {results/EE_UE_TS_q2_proposed_smooth.dat};
        \addlegendentry{TS, $q=2$}

        \addplot[
        teal, mark=triangle*, mark repeat=40, mark size=1.6,
        line width=1.1pt
        ] table {results/EE_UE_ES_q1_proposed_smooth.dat};
        \addlegendentry{ES, $q=1$}

        \addplot[
        magenta, mark=diamond*, mark repeat=40, mark size=1.6,
        line width=1.1pt
        ] table {results/EE_UE_ES_q2_proposed_smooth.dat};
        \addlegendentry{ES, $q=2$}

        \addplot[blue, dashed, line width=1.1pt]
        table {results/EE_UE_TS_q1_badiu_smooth.dat};
        \addlegendentry{\cite{badiu2020communication}, TS, $q=1$}

        \addplot[red, dashed, line width=1.1pt]
        table {results/EE_UE_TS_q2_badiu_smooth.dat};
        \addlegendentry{\cite{badiu2020communication}, TS, $q=2$}

        \addplot[teal, dashed, line width=1.1pt]
        table {results/EE_UE_ES_q1_badiu_smooth.dat};
        \addlegendentry{\cite{badiu2020communication}, ES, $q=1$}

        \addplot[magenta, dashed, line width=1.1pt]
        table {results/EE_UE_ES_q2_badiu_smooth.dat};
        \addlegendentry{\cite{badiu2020communication}, ES, $q=2$}

        \end{axis}
        \end{tikzpicture}
        \caption{}
        \label{fig:EEUS}
    \end{subfigure}

    \caption{Energy efficiency versus the number of reflecting elements for: (a) Tx-side zeRIS, (b) UE-side zeRIS.}
    \label{fig:EEanalysis}
\end{figure}

Finally, Fig.~\ref{fig:EEanalysis} illustrates the energy efficiency of a zeRIS-assisted network as a function of the number of reflecting elements $N$, for both Tx-side and UE-side deployment scenarios, considering the TS and ES HaR schemes under quantized phase shifts with resolutions $q=1$ and $q=2$, while also comparing the proposed analytical framework with the benchmark results of~\cite{badiu2020communication}. It can be observed that, across all examined cases, the ES scheme consistently achieves higher energy efficiency than the TS scheme, indicating its superior capability in balancing energy harvesting and communication processes. At the same time, the case of $q=1$ emerges as the most energy-efficient configuration compared to $q=2$, despite the improved coherent combining offered by higher phase resolution, which reveals a non-trivial trade-off where the additional power consumption required for finer quantization outweighs the corresponding communication gains. This behavior is consistently observed in both deployment scenarios, although it becomes more pronounced in the UE-side case due to the increased reliance on the energy harvesting link. Notably, discrepancies between the proposed model and the benchmark results are clearly evident, as in the Tx-side deployment the use of~\cite{badiu2020communication} would lead to the conclusion that TS is more energy efficient than ES for $q=1$, while in the UE-side deployment the benchmark model suggests that the TS scheme with $q=1$ yields zero energy efficiency over the entire examined range of $N$, whereas the proposed model indicates that the optimal operation is attained at $N=934$ elements. These observations underline that phase quantization in zeRIS-assisted systems jointly affects energy harvesting and communication in a coupled manner, leading to non-intuitive design regimes, while also demonstrating that inaccurate modeling of the residual phase error may result in misleading conclusions regarding both the optimal HaR strategy and the required number of reflecting elements for energy-efficient operation.

\section{Conclusion} 
In this work, we developed a comprehensive analytical framework to characterize the performance of zeRIS-assisted communication systems under quantized phase control and energy-constrained operation. By explicitly modeling the residual phase error induced by quantization and incorporating its statistical properties into the analysis, we derived tractable expressions for the joint energy–data rate outage probability under both TS and ES schemes and for different deployment scenarios. The obtained results revealed that quantization affects zeRIS operation in a coupled manner, simultaneously influencing energy harvesting and signal reflection and inducing non-trivial trade-offs, while also showing that increasing the phase resolution does not necessarily improve performance. Furthermore, the analysis demonstrated that simplified or inaccurate modeling of residual phase errors may lead to misleading conclusions regarding both the achievable performance and the optimal system design, particularly in regimes where energy constraints dominate. Therefore, the presented results provide clear design guidelines for zeRIS-assisted systems, emphasizing the need to jointly account for quantization and propagation conditions when selecting the phase resolution and the applied HaR scheme.

\appendices
\section{Proof of Proposition \ref{prop:4}}
Considering the definition of the joint energy-data rate outage probability, it follows that for the Tx-side deployment and the TS scheme it can be expressed as
\begin{equation}\label{PGEN_TS_BS}
\small
\begin{split}
P_{\mathrm{B}}^{\mathrm{TS}} 
= &\Pr \Bigg( \tau T \zeta P_t G_t \ell_1 
\left|\sum_{i=1}^{N} e^{j\left(\omega_i^{\mathrm{h}}+\arg(h_{1i})\right)}\right|^2 \\
&\leq T\big((1-\tau)NP_{\mathrm{elem}}+P_{\mathrm{ctrl}}\big) \\
&\cup \ (1-\tau)\log_2\!\left(1+\gamma_t G \ell_1 \ell_2 
\left|\sum_{i=1}^{N}|h_{2i}|e^{j\varepsilon_i}\right|^2\right) \leq R_{\mathrm{thr}} \Bigg).
\end{split}
\end{equation}
During the energy harvesting phase, the zeRIS is configured to align the impinging signals by selecting the phase shifts as \(\omega_i = \mathbb{Q}(-\arg(h_{1i}))\). Since the Tx-zeRIS link is purely LoS, it holds that \(\arg(h_{1i}) = \frac{2\pi d_1}{\lambda}\) for all the zeRIS elements. As a result, even in the presence of phase quantization, all elements apply the same quantized phase shift, which leads to a common residual phase. Therefore, the reflected signals remain fully coherent during the energy harvesting phase. Consequently, $\left|\sum_{i=1}^{N} e^{j\left(\omega_i^{\mathrm{h}}+\arg(h_{1i})\right)}\right|^2$ reduces to  $N^2$, and thus the harvested energy becomes independent of the residual phase error. Therefore, the energy outage event can be equivalently written as \(P_{\mathrm{B}}^{\mathrm{TS}}=1\) when \(N \leq N_{\min}^{\mathrm{TS}}\).

For the case where \(N > N_{\min}^{\mathrm{TS}}\), the outage event is determined only by the rate condition. In particular, by defining
\begin{equation}
\small
X=\left|\sum_{i=1}^{N}|h_{2i}|e^{j\varepsilon_i}\right|^2,
\end{equation}
the rate outage event can be rewritten as $X \leq x^B_{\mathrm{TS}}$, which yields
\begin{equation}
\small
P_{\mathrm{B}}^{\mathrm{TS}} = \Pr\!\left(X \leq x^B_{\mathrm{TS}}\right).
\end{equation}
To proceed, we approximate the RV \(X\) through moment matching by a Gamma distributed RV with parameters
\begin{equation}
\small
k_{\mathrm{TS}} = \frac{\mathbb{E}^2[X;N,d_2]}{\mathbb{E}[X^2;N,d_2]-\mathbb{E}^2[X;N,d_2]},
\end{equation}
and
\begin{equation}
\small
\theta_{\mathrm{TS}} = \frac{\mathbb{E}[X^2;N,d_2]-\mathbb{E}^2[X;N,d_2]}{\mathbb{E}[X;N,d_2]}.
\end{equation}
To calculate \(\mathbb{E}[X;N,d_2]\) and \(\mathbb{E}[X^2;N,d_2]\), we set $X_i = |h_{2i}|e^{j\varepsilon_i}$, and $S=\sum_{i=1}^{N}X_i$, such that \(X=|S|^2=SS^{*}\) \cite{ThrassosHQAM}. After some algebraic manipulations, \(X\) and \(X^2\) can be expressed as
\begin{equation}\label{eq:X_expand_new}
\small
X=\sum_{i=1}^{N}\sum_{k=1}^{N}X_iX_k^{*},
\end{equation}
and
\begin{equation}\label{eq:X2_expand_new}
\small
X^2=\sum_{i=1}^{N}\sum_{k=1}^{N}\sum_{p=1}^{N}\sum_{q=1}^{N}X_iX_k^{*}X_pX_q^{*}.
\end{equation}
By observing \eqref{eq:X_expand_new} and \eqref{eq:X2_expand_new}, it becomes evident that \(\mathbb{E}[X;N,d_2]\) and \(\mathbb{E}[X^2;N,d_2]\) consist of \(N^2\) and \(N^4\) terms, respectively. Since expectation is a linear operator, both moments are obtained by summing the mean values of the corresponding terms after grouping together those with the same index pattern. In particular, for \(\mathbb{E}[X;N,d_2]\), the summation terms in \eqref{eq:X_expand_new} are the following:
\begin{itemize}
    \item \(N\) terms of the form \(|X_i|^2=|h_{2i}|^2\), when \(i=k\),
    \item \(N^2-N\) terms of the form \(X_iX_k^{*}\), when \(i\neq k\).
\end{itemize}
Therefore, \(\mathbb{E}[X;N,d_2]\) can be expressed as
\begin{equation} \label{final_ex}
\small
\begin{split}
\mathbb{E}[X;N,d_2]
&= N \mathbb{E}[|h_2|^2] + (N^2-N)\mathbb{E}[X_iX_k^{*}] \\
&= N \mathbb{E}[|h_2|^2] + N(N-1)\mathbb{E}[|h_2|]^2 |\mu_1(d_2)|^2,
\end{split}
\end{equation}
where the second equality follows from the independence among the fading amplitudes and the residual phase errors, together with \(\mu_1(d_2)=\mathbb{E}[e^{j\varepsilon}]\). Similarly, for \(\mathbb{E}[X^2]\), the expansion in \eqref{eq:X2_expand_new} leads to terms with different index multiplicities, namely, terms with one repeated index four times, terms with two distinct indices, terms with three distinct indices, and terms with four distinct indices. Since expectation is a linear operator, \(\mathbb{E}[X^2]\) is obtained by summing the mean value of each class after counting its multiplicity. More specifically, the non-zero contributions are the following:
\begin{itemize} 
\item \(N\) terms of the form \(X_iX_i^{*}X_iX_i^{*}=|X_i|^4\), yielding \(\mathbb{E}[|h_2|^4]\). 
\item \(4N(N-1)\) terms corresponding to the \(3+1\) index pattern, i.e., terms in which one index appears three times while another appears once, as exemplified by \(X_iX_i^{*}X_iX_k^{*}\), with \(i\neq k\). In this case, $X_iX_i^{*}X_iX_k^{*} = |h_{2i}|^2|h_{2i}||h_{2k}|e^{j(\varepsilon_i-\varepsilon_k)}$, and thus its expectation is given by
\[
\mathbb{E}[|h_2|^2]\mathbb{E}[|h_2|]^2|\mu_1(d_2)|^2.
\]
\item \(N(N-1)\) terms corresponding to the \(2+2\) pattern of the form \(X_i^2(X_k^{*})^2\), \(i\neq k\). In this case, \[ X_i^2(X_k^{*})^2 = |h_{2i}|^2|h_{2k}|^2 e^{j2(\varepsilon_i-\varepsilon_k)}, \] and therefore its expectation is \[ \mathbb{E}[|h_2|^2]^2|\mu_2(d_2)|^2. \] 
\item \(2N(N-1)\) terms corresponding to the \(2+2\) pattern of the form \(|X_i|^2|X_k|^2\), \(i\neq k\), yielding directly \[ \mathbb{E}[|h_2|^2]^2. \] 
\item \(N(N-1)(N-2)\) terms corresponding to the \(2+1+1\) pattern. A representative term is \(X_iX_i^{*}X_kX_p^{*}\), with \(i\), \(k\), and \(p\) pairwise distinct, which gives \[ |h_{2i}|^2|h_{2k}||h_{2p}|e^{j(\varepsilon_k-\varepsilon_p)}. \] In addition, terms such as \(X_iX_k^{*}X_iX_p^{*}\) and \(X_iX_k^{*}X_pX_i^{*}\) also arise under the same multiplicity pattern. By evaluating the expectation of all such terms and combining conjugate pairs, the total contribution of this class becomes \[ \mathbb{E}[|h_2|^2]\mathbb{E}[|h_2|]^2 \Big(2\mathrm{Re}\!\{\mu_2(d_2)(\mu_1^{*}(d_2))^2\}+4|\mu_1(d_2)|^2\Big). \]
\item \(N(N-1)(N-2)(N-3)\) terms corresponding to the all-distinct pattern, represented by \(X_iX_k^{*}X_pX_q^{*}\), where \(i\), \(k\), \(p\), and \(q\) are all different. By independence, the corresponding expectation factorizes as \[ \mathbb{E}[|h_2|]^4|\mu_1(d_2)|^4. \]
\end{itemize}
As a result, \(\mathbb{E}[X^2]\) is obtained as 
\begin{equation}\label{final_exx} 
\small
\begin{split} 
&\mathbb{E}[X^2] = N\mathbb{E}[|h_2|^4] + 4N(N-1)\mathbb{E}[|h_2|^2]\mathbb{E}[|h_2|]^2|\mu_1(d_2)|^2 \\
&\quad + N(N-1)\mathbb{E}[|h_2|^2]^2|\mu_2(d_2)|^2 + 2N(N-1)\mathbb{E}[|h_2|^2]^2 \\ 
&\quad + N(N-1)(N-2)\mathbb{E}[|h_2|^2]\mathbb{E}[|h_2|]^2 \\ 
&\qquad \times \Big(2\mathrm{Re}\!\{\mu_2(d_2)(\mu_1^{*}(d_2))^2\}+4|\mu_1(d_2)|^2\Big) \\ &\quad + N(N-1)(N-2)(N-3)\mathbb{E}[|h_2|]^4|\mu_1(d_2)|^4. 
\end{split} 
\end{equation} 
Finally, since \(|h_2|\) follows a Nakagami-\(m\) distribution, we can obtain the moments of \(|h_2|\) by \cite{ThrassosHQAM} and by substituting them into \eqref{final_ex} and \eqref{final_exx}, then \eqref{eq:EX_TS} and \eqref{eq:EX2_TS} can be acquired. Thus, by employing the CDF of the Gamma distribution, we obtain \eqref{PBTS}, which concludes the proof.

\bibliographystyle{IEEEtran}
\bibliography{reference}

\end{document}